\documentclass[11pt]{amsart}
\usepackage[margin=1in]{geometry}
\usepackage{amsfonts, float, mathtools}
\usepackage{amssymb}
\usepackage[dvips]{graphics}
\usepackage{epsfig}
\usepackage{euscript}
\usepackage{color}
\usepackage{cite}

\usepackage[all]{xy}
\usepackage{multirow}
\usepackage{tikz}

\usepackage{fancyvrb}

\usepackage[pdftex]{hyperref}

 \newtheorem{thm}{Theorem}[section]
 \newtheorem{cor}[thm]{Corollary}
 \newtheorem{lemma}[thm]{Lemma}
 \newtheorem{prop}[thm]{Proposition}

 \theoremstyle{definition}

 \theoremstyle{remark}
 
 \newtheorem{assumption}[thm]{Assumption}
 \numberwithin{equation}{section}

 \def\idty{{\mathchoice {\mathrm{1\mskip-4mu l}} {\mathrm{1\mskip-4mu l}} %
{\mathrm{1\mskip-4.5mu l}} {\mathrm{1\mskip-5mu l}}}}

\newcommand{\bR}{{\mathbb R}}

\newcommand{\bC}{{\mathbb C}}

\newcommand{\bN}{{\mathbb N}}
\newcommand{\bZ}{{\mathbb Z}}

\newcommand{\cB}{{\mathcal B}}

\newcommand{\cG}{{\mathcal G}}
\newcommand{\cH}{{\mathcal H}}
\newcommand{\cI}{{\mathcal I}}
\newcommand{\cT}{{\mathcal T}}

\newcommand{\cS}{{\mathcal S}}
\newcommand{\cE}{{\mathcal E}}

\newcommand{\caB}{{\mathcal B}}
\newcommand{\caC}{{\mathcal C}}
\newcommand{\caD}{{\mathcal D}}

\newcommand{\caG}{{\mathcal G}}
\newcommand{\caH}{{\mathcal H}}
\newcommand{\caI}{{\mathcal I}}

\newcommand{\cK}{{\mathcal K}}

\newcommand{\caR}{{\mathcal R}}
\newcommand{\caS}{{\mathcal S}}
\newcommand{\caT}{{\mathcal T}}

\newcommand{\caV}{{\mathcal V}}

\newcommand{\bbC}{{\mathbb C}}

\newcommand{\bbZ}{{\mathbb Z}}

\newcommand{\gap}{\mathrm{gap}}

\newcommand{\braket}[2]{\left\langle #1 , #2\right\rangle}

\newcommand{\Span}{\mathrm{span}}
\newcommand{\ran}{{\rm ran}}

\newcommand{\ketbra}[1]{\vert #1\rangle\langle #1\vert}
\newcommand{\kettbra}[2]{\vert #1\rangle\langle #2\vert}

\newcommand{\ket}[1]{\vert #1 \rangle}

\newcommand{\spec}{\mathop{\rm spec}}

\newcommand{\be}{\begin{equation}}
\newcommand{\ee}{\end{equation}}
\newcommand{\bea}{\begin{eqnarray}}
\newcommand{\eea}{\end{eqnarray}}
\newcommand{\beann}{\begin{eqnarray*}}
\newcommand{\eeann}{\end{eqnarray*}}

\newcommand{\eq}[1]{(\ref{#1})}

\title[]{On a bulk gap strategy for quantum lattice models}

\author[A. Young]{Amanda Young}
\address{Munich Center for Quantum Science and Technology and Zentrum Mathematik, TU M\"{u}nchen 85747 Garching, Germany}
\address{Department of Mathematics, University of Illinois Urbana-Champaign, Urbana, IL, USA 61801}

\begin{document}
\date{\today }

\maketitle

\begin{abstract}
Establishing the (non)existence of a spectral gap above the ground state in the thermodynamic limit is one of the fundamental steps for characterizing the topological phase of a quantum lattice model. This is particularly challenging when a model is expected to have low-lying edge excitations, but nevertheless a positive bulk gap. We review the bulk gap strategy introduced in \cite{Warzel:2022, Warzel:2023} while studying truncated Haldane pseudopotentials. This approach is able to avoid low-lying edge modes by separating the ground states and edge states into different invariant subspaces before applying spectral gap bounding techniques. The approach is stated in a general context, and we reformulate specific spectral gap methods in an invariant subspace context to illustrate the necessary conditions for combining them with the bulk gap strategy. We then review its application to a truncation of the 1/3-filled Haldane pseudopotential in the cylinder geometry.
\end{abstract}

\section{Introduction}\label{sec:intro}
One of the fundamental quantities in the classification of quantum phases of matter, including topological phases, is the spectral gap \cite{Bachmann2012, Chen2010, Chen2011, Hastings2005, Ogata2021}, and a model belongs to a gapped ground state phase if the spectrum of its GNS Hamiltonian in the thermodynamic limit has a positive gap above its ground state energy. This is typically proved by showing that the model is \emph{uniformly gapped}, meaning that the spectral gaps of an appropriately chosen sequence of local Hamiltonians can be bounded from below by a positive constant independent of the system size.

When the physical space of the infinite system does not have a boundary, e.g. $\bZ^d$, the gap of the GNS Hamiltonian can also be referred to as a \emph{bulk gap}. In such cases, imposing different boundary conditions on the local Hamiltonians can result in the same GNS Hamiltonian, see \cite[Section III]{Nachtergaele2019}. One can then try to optimize the uniform gap over different boundary conditions to produce a sharper lower bound on the bulk gap. This is especially useful when the model with open boundary conditions has low-lying (or even gapless) edge modes. These are low energy states where the excitations are localized to the boundary of the system. As the boundary disappears in the thermodynamic limit, these do not typically converge to bulk excitations, and so the model with periodic boundary conditions would be better suited for studying the bulk gap.

There are two main classes of methods for proving uniform gap estimates for frustration-free quantum spin systems: ones based on using ground state projections to localize low-lying excitations \cite{Affleck1988, Anshu2020, Fannes1992, Kastoryano2019, Nachtergaele1996, Nachtergaele2018}, and finite size criteria which prove uniform gaps in the situation that the spectral gap of finite volume Hamiltonian is sufficiently large \cite{Knabe1988, Gosset2016, Lemm2019, Lemm2020}. Unfortunately, the main idea that makes both types of methods successful also makes them susceptible to edge modes, in the sense that the uniform bound produced would be proportional to the energy of the edge states regardless of the chosen boundary conditions. In cases where the edge state energy is anticipated to be significantly smaller than the bulk gap, or even vanishing in the thermodynamic limit, this results in uniform gap estimates that do not accurately reflect the behavior of the bulk gap.

The aim of this work is to review recent progress in proving bulk gaps in the presence of edge modes. The spectral gap for truncatations of Haldane pseudopotentials with certain fillings on different geometries was analyzed in \cite{Nachtergaele2021, Nachtergaele2020, Warzel:2022, Warzel:2023}. These models have edge modes that are nonvanishing, but small, which produce uniform gap estimates that are unstable in the limit of a certain model parameter and do not reflect the true behavior of the bulk gap. To overcome these low-lying excitations, a new scheme based off identifying invariant subspaces of the local Hamiltonian and catering the gap method to the individual subspaces was introduced in \cite{Warzel:2022, Warzel:2023}. We explain the bulk gap strategy in a general context as it can be applied to other models, and illustrate its application to a truncation of the 1/3-filled Haldane pseudopotential on the cylinder. 

\subsection{The truncated $1/3$-filled Haldane pseudopotential}
Haldane pseudopotentials were first introduced in \cite{Haldane1983} as Hamiltonian models for the fractional quantum Hall effect (FQHE). These are designed to have a Laughlin state \cite{Laughlin1983} as a maximally filled ground state, and it has long been conjectured that they exhibit other characteristic properties of the FQH phase, including a spectral gap above its ground state energy \cite{Haldane1985, Pokrovsky1985, Trugman1985}. While numerical evidence supports this conjecture \cite{Haldane1990}, a rigorous proof has been elusive.

Letting $\alpha=\ell/R$ denote the ratio of the magnetic length to the cylinder radius, the Hilbert space for the $1/3$-filled model on an infinite cylinder is the fermionic Fock space generated by the lowest Landau orbitals:
\[
\psi_n(x,y) = \left(\frac{\alpha}{2\pi^{3/2}}\right)^{1/2}e^{i\frac{\alpha}{\ell}ny}e^{-\frac{1}{2}(\frac{x}{\ell}-\alpha n)^2}, \qquad n \in \bbZ,
\]
where $x\in\bR$ and $y\in[0,2\pi R).$ Denoting by $c_n$ the annihilation operator associated with $\psi_n,$ the $1/3$-filled Haldane pseudopotential in second quantization takes the form
\begin{equation}\label{eq:Haldane}
    W\propto \sum_{s\in\bbZ/2}B_s^*B_s, \quad B_s = \sum_{k\in \bZ+s}F(2\alpha k)c_{s-k}c_{s+k}, \quad F(t) = te^{-t^2/4}.
\end{equation}
One challenge for proving the spectral gap conjecture is that the above quantum lattice interaction is not finite range. This motivated \cite{Nachtergaele2021}, where the gap question was rigorously treated for a finite-range truncation of the interaction that should well-approximate the model for small cylinder radii. This truncated model, first introduced in \cite{Nakamura2012} and subsequently analyzed in \cite{Wang2015}, is defined by restricting $B_s$ to the sum over $|k|\leq 3/2.$ 

Apply the Jordan-Wigner transformation leads to a quantum spin system with local Hilbert space $\caH_\Lambda = (\bC^2)^{\otimes|\Lambda|}$ for any finite interval $\Lambda=[a,b]\subseteq \bZ.$ Generalizing the model parameters to $\kappa>0$ and $\lambda\in\bbC$, the local Hamiltonian with open boundary conditions (OBC) for the truncated $1/3$-filled Haldane pseudopotential becomes
\begin{align}
    & H_\Lambda = \sum_{x=a}^{b-2}n_xn_{x-2} + \kappa \sum_{x=a+1}^{b-2}q_x^*q_x\quad \text{where} \label{eq:OBC_Ham} \\
  n_x = \sigma_x^+\sigma_x^-, &\quad q_x = \sigma_{x}^-\sigma_{x+1}^--\lambda \sigma_{x-1}^-\sigma_{x+2}^-, \quad \sigma^-=(\sigma^+)^* =\kettbra{0}{1}
\,.
   \label{eq:interaction}
\end{align}
where $n\ket{i}=i\ket{i}$ labels an eigenbasis of $n=\sigma^+\sigma^-\in\caB(\bC^2)$. We say that $\ket{0}$ denotes a vacant site, and $\ket{1}$ denotes an occupied site. The Hamiltonian preserves particle number $N_\Lambda = \sum_xn_x$, and center of mass $M_\Lambda =\sum_{x}xn_x,$ and the choices 
\begin{equation}
\kappa=(F(\alpha)/F(2\alpha))^2=e^{3\alpha^2/2}/4,\qquad \lambda=-F(3\alpha)/F(\alpha)=-3e^{-2\alpha^2}    
\end{equation}
correspond to the \emph{physical regime} in \eqref{eq:Haldane}.

The model with OBC is uniformly gapped if there is a sequence of finite intervals $\Lambda_n \uparrow \bZ$ so that
\[
\gamma := \liminf_{n\to\infty}\gap(H_{\Lambda_n})>0,
\]
where $\gap(H_{\Lambda_n})$ is the difference between the first excited state and ground state energies. While the uniform gap is independent of the system size $\Lambda$, it does depend on the model parameters, $\gamma=\gamma(\kappa,\lambda).$ The lower bound obtained in \cite{Nachtergaele2021} for the uniform gap of \eqref{eq:OBC_Ham} satisfied $\gamma \geq \mathcal{O}(|\lambda|^2)$ in the regime of small $|\lambda|$. This corresponds to $R$ small in the physical regime. It was also show that this estimate is sharp. For example, $\Span\{\ket{110010\ldots 0}, \, \ket{10110\ldots 0}\}$ is invariant under $H_\Lambda$ and has an eigenvalue $\frac{\kappa}{\kappa+1}|\lambda|^2+\mathcal{O}(|\lambda|^4)$ for small $|\lambda|$. As the ground state space of this model is $\caG_\Lambda = \ker(H_\Lambda)$, this is an upper bound on the spectral gap. However, in the thermodynamic limit $\Lambda\uparrow \bZ$, any state from from this subspace converges to the the vacuum state, which is a ground state for the generator of the the infinite system dynamics. This is the hallmark of an edge excitation.

Numerical results for the model with periodic boundary conditions (PBC) further showed that the spectral gap had a positive limit as $|\lambda|\to 0$ for small system sizes  \cite{Nachtergaele2021}. This suggests that the eigenstates of the model with OBC and $\mathcal{O}(|\lambda|^2)$ energy are edge states, and that the uniform gap
\begin{equation}\label{eq:bulk_gap_1}
    \liminf_{n\to\infty} \gap(H_{\Lambda_n}^{\rm per} )>0
\end{equation}
of the model with PBC would produce better bounds on the bulk gap for small $|\lambda|$ where
\begin{equation}\label{eq:PBC_Ham}
    H_\Lambda^{\rm per} = \sum_{x=a}^{b}(n_{x}n_{x+2} + \kappa q_x^*q_x).
\end{equation}
(Note that for the periodic interaction, addition is understood modulo $|\Lambda|.$) The main difficulty to proving this claim is that existing methods for proving positive uniform gaps rely on spectral gaps for local Hamiltonians associated with subvolumes $\Lambda'\subseteq\Lambda$. Even if one considers the model with PBC, the natural boundary conditions for a subvolume are OBC. Thus, if $H_{\Lambda'}$ has edge modes, this will be reflected in the bound for $\gap(H_\Lambda^{\rm per})$, producing an estimate on the bulk gap with the wrong behavior in the limit $|\lambda|\to 0.$ The bulk gap approach from \cite{Warzel:2022, Warzel:2023} overcame this challenge. In sections \ref{sec:subspaces}-\ref{sec:gap_proof}, we apply this strategy to \eqref{eq:PBC_Ham} to produce the following result.


\begin{thm}
    \label{thm:main_result}
    Fix any nonzero $\lambda \in \bC$ and $\kappa>0.$ There exists a monotone increasing function $f:[0,\infty)\to[0,\infty)$ such that if $f(|\lambda|^2)<1/3$, then
    \begin{equation}
        \liminf_{L\to\infty}\gap(H_{[1,L]}^{\rm per})\geq \min\left\{\frac{\kappa}{6(1+2|\lambda|^2)}\left(1-\sqrt{3f(|\lambda|^2)}\right)^2, \, \gamma^{\rm per}\right\}
    \end{equation}
    where
    \begin{equation}
        \gamma^{\rm per}= \frac{1}{3}\min\left\{1, \frac{\kappa}{2+2\kappa|\lambda|^2}, \, \frac{\kappa}{\kappa+2}, \right\}.
    \end{equation}
\end{thm}
The function $f(r)$, which is defined in Section~\ref{sec:E_1}, was analyzed in \cite[Appendix A]{Nachtergaele2021}, where it was shown that $f(|\lambda|^2)<1/3$ for all $|\lambda|<5.3$. This range includes the entire physical regime.

While we focus on the truncated $1/3$-filled model on the cylinder, this strategy has also been successfully applied to the analogously truncated $1/3$-filled model on the torus \cite{Warzel:2022}, as well as a truncation of the $1/2$-filled model on the cylinder \cite{Warzel:2023}, the latter of which is of interest for studying rapidly rotating Bose gases \cite{Lewin2009, Seiringer2020}.

In Section~\ref{sec:gap_strategy}, we explain the bulk gap strategy and prove the two spectral gap techniques used in its application to the Haldane pseudopotential. In contrast to previous statements of these results, e.g. \cite{Knabe1988, Nachtergaele1996}, we formulate them in an invariant subspace context to illustrate the necessary conditions for using them with the bulk gap strategy. In Section~\ref{sec:subspaces}, we identify the invariant subspaces of the truncated Haldane pseudopotential used for the bulk gap strategy, and discuss some of their important properties. In Section~\ref{sec:gap_proof} the bulk gap strategy is applied to prove Theorem~\ref{thm:main_result}. We conclude with a discussion of the edge states. Throughout this text, we use $H_\Lambda^{\rm per}$ and $H_\Lambda$ to denote a local Hamiltonian with PBC and OBC, respectively.

\section{Spectral gap techniques and the bulk gap strategy} \label{sec:gap_strategy}

\subsection{The invariant subspace approach} 

 The main idea underlying all spectral gap methods for quantum spin systems comes from the intuition that low-lying bulk excitations of a short-ranged model should be localized to finite regions of space. Roughly, if one determines the largest size $M$ that a region needs to be to effectively localize such an excitation, then the spectral gap of the Hamiltonian associated to a region of size $M$ can be used to produce a lower bound on the spectral gap of the Hamiltonian associated with any (sufficiently) larger region $\Lambda$. The type of estimate this produces is approximately of the form:
\begin{equation}
\gap(H_\Lambda^{\#}) \gtrapprox C^{\#} \inf_{\substack{X\subseteq \Lambda : \\|X| \leq M}}\gap(H_X), \quad \#\in \{{\rm per}, \, \cdot\}
\end{equation}
where $C^\#>0$ is independent of $\Lambda,$ but depends on $M.$ Rigorous statements will be made precise and explicit in Sections~\ref{sec:fsc}-\ref{sec:MM}. As the Hamiltonian for the subregion where an excitation is localized naturally has OBC, if $H_X$ has low-lying edge modes, the lower bounds produced for $H_\Lambda^{\rm per}$ will reflect these energies, even though they do not persist to the model with PBC.

The idea behind the invariant subspace strategy is to separate the ground states and edge states into different invariant subspaces, so that the edge modes will not be detected in the gap estimates. The trade off is that one needs to also produce a lower bound on a ground state energy. To be more explicit, if $\cG_\Lambda^{\rm per}$ is the ground state space of $H_\Lambda^{\rm per}$, and $\cE_{\Lambda}$ is the space of edge states associated with $H_\Lambda$, one wants to decompose the Hilbert space as
\[
\cH_\Lambda = \caV_\Lambda \oplus \caV_\Lambda^\perp, \quad H_\Lambda^{\rm per}\caV_\Lambda \subseteq \caV_\Lambda, \quad \cG_\Lambda^{\rm per}\subseteq \caV_\Lambda, \quad \cE_{X}\otimes \cH_{\Lambda\setminus X}\subseteq \caV_\Lambda^\perp, \quad \forall X\subseteq \Lambda.
\]
In the case that the ground state energy of $H_\Lambda^{\rm per}$ is zero, the spectral gap of $H_\Lambda^{\rm per}$ is then
\begin{equation}\label{eq:bulk_gap}
    \gap(H_\Lambda^{\rm per}) = \min\{E_{1}(\caV_\Lambda), \, E_0(\caV_\Lambda^\perp)\}
\end{equation}
where
\begin{equation}\label{E_defs}
  E_1(\caV_\Lambda) = \inf_{0\neq \psi \in \caV_\Lambda\cap (\caG_\Lambda^{\rm per})^\perp}\frac{\braket{\psi}{H_\Lambda^{\rm per}\psi}}{\|\psi\|^2}, \quad E_0(\caV_\Lambda^\perp) = \inf_{0\neq \psi \in \caV_\Lambda^\perp}\frac{\braket{\psi}{H_\Lambda^{\rm per}\psi}}{\|\psi\|^2}.  
\end{equation}
The first quantity, $E_1(\caV_\Lambda)$, is the spectral gap of $H_\Lambda^{\rm per}\restriction_{\caV_\Lambda}$, and applying gap methods to this restricted Hamiltonian will produce an estimate that is unimpeded by the edge modes. The second quantity, $E_0(\caV_\Lambda^\perp)$, is the ground state energy of $H_\Lambda^{\rm per}\restriction_{\caV_\Lambda^\perp}$. As long as one determines a lower bound on this quantity that is independent of the ground state energy of $H_X\restriction_{\caV_\Lambda^\perp}$ for any $X\subseteq \Lambda$, the result is a lower bound on $\gap(H_\Lambda^{\rm per})$ that is independent of the edge excitations.

Let us briefly point out that an obvious candidate for $\caV_\Lambda$ is $\caG_\Lambda^{\rm per}$ itself. However, the orthogonal complement of the ground state space will often have a complicated description that is not easy to manipulate. Taking a larger space for $\caV_\Lambda$ might yield an orthogonal complement with a simple description that is also easily seen to support the edge states. This is the case for the truncated Haldane pseudopotential, where $\caV_\Lambda^\perp$ will have an orthonormal basis of configuration states:
\[
\caV_\Lambda^{\perp} = \Span\{\ket{\mu} : \mu \in \caI_\Lambda\}, \quad \caI_\Lambda \subset \{0,1\}^{|\Lambda|}.
\]

In Section~\ref{sec:gap_methods}, we present proofs of the gap methods that are used to bound $E_1(\caV_\Lambda)$ for the truncated Haldane pseudopotential in Section~\ref{sec:E_1}. The approach used to bound $E_0(\caV_\Lambda^\perp)$ for this model uses Cauchy-Schwarz estimates to show
\[
\braket{\psi}{H_\Lambda^{\rm per}\psi} = \sum_{\mu\in \cI_\Lambda}\sum_{x\in \Lambda}\mu_{x}\mu_{x+2}|\psi(\mu)|^2+\sum_{\nu\in\{0,1\}^{|\Lambda|}}\sum_{x\in \Lambda}|\braket{\nu}{q_x\psi}|^2\geq \gamma^{\rm per} \sum_{\mu\in \caI_\Lambda}|\psi(\mu)|^2
\]
for any $\psi=\sum_{\mu\in \caI_\Lambda}\psi(\mu)\ket{\mu}\in\caV_\Lambda^\perp$. This is explained in more detail in Section~\ref{sec:E_0}.

\subsection{Spectral gap methods}\label{sec:gap_methods}

Both finite size criteria and ground state projection methods can be applied to systems with open or periodic boundary conditions \cite{Gosset2016,Lemm2019,Nachtergaele1996,Young2016}. However, finite size criteria are more amenable to PBC, while ground state projection methods, such as the martingale method, work well for models with OBC. The proof of Theorem~\ref{thm:main_result} follows this tradition - a version of Knabe's finite size criterion \cite{Knabe1988} is used to produce a $\Lambda$-independent lower bound on $E_1(\caV_\Lambda)$ for the truncated Haldane pseudopotential. Finite size criteria require that the spectral gap of a finite size Hamiltonian with OBC is sufficiently large. To complete the proof of the bulk gap, we use a martingale method \cite{Nachtergaele2018} to show that the gap of this Hamiltonian is larger than the necessarily threshold. However, these are not the only gap methods that can be used with this bulk gap strategy. For example, the technique introduced in \cite{Fannes1992} for finitely correlated states and generalized to decorated lattice models in \cite{AbdulRahman2020, Lucia:2023} can also be adapted to this invariant subspace approach.

\subsubsection{A coarse-grained finite size criterion}\label{sec:fsc}
We begin with presenting a generalized version of Knabe's finite size criterion. For this result, let $\cH$ denote an arbitrary finite-dimensional Hilbert space, and assume that $\{P_i: 1\leq i \leq N\}$ is a family of orthogonal projections on $\cH$ that satisfy the following commutation relations:
\begin{equation} \label{eq:commutator}
    [P_i,P_j] \neq 0 \implies |i-j|=1 \; \text{or}\; \{i,j\} = \{1,N\}.
\end{equation}
With these projections and for $1\leq n,k\leq N$, define Hamiltonians
\begin{equation}\label{eq:fcs_hams}
H_{N} = \sum_{i=1}^N P_i, \quad H_{n,k} =\sum_{i=k}^{n+k-1}P_i
\end{equation}
where $i \equiv i+N.$ Note that $H_N\geq 0$ as it is the sum of non-negative terms. The finite size criterion shows that if the ground state energy of $H_N$ is zero, then $\gap(H_N)$ can be bounded from below by a constant that only depends on $n$ and $\inf_k\gap(H_{n,k})$.

\begin{thm}[Generalized Knabe Bound \cite{Knabe1988}]\label{thm:Knabe} Suppose $H_N$ and $H_{n,k}$ are defined as in \eqref{eq:fcs_hams} for orthogonal projections that satisfy \eqref{eq:commutator}. If $\ker(H_N)\neq \{0\}$, then for any integer $1<n\leq N/2$,
\begin{equation}\label{eq:Knabe_bound}
  \gap(H_N) \geq \frac{n}{n-1}\left(\min_{1\leq k \leq N}\gap(H_{n,k})-\frac{1}{n}\right).
\end{equation}
\end{thm}

\begin{proof}
The result will follow from using that the Hamiltonians $H_{n,k}$, $1\leq k \leq N$, satisfy
    \begin{equation} \label{Hnk_properties}
            H_{n,k}^2 \geq \gap(H_{n,k})H_{n,k}, \qquad \sum_{k=1}^{N}H_{n,k} = n H_N 
    \end{equation}
to show that $H_N^2 \geq \gamma_nH_N$ where $\gamma_n$ is the lower bound in \eqref{eq:Knabe_bound}. As addition is taken modulo $N$, $H_N^2$ and $H_{n,k}^2$ can be rewritten as
\begin{align}
     H_N^2  &= \sum_{i=1}^N P_i + \sum_{i<j}\{P_i,P_j\} = H_N + \sum_{m=1}^{\lfloor N/2\rfloor}\sum_{i=1}^{N}\{P_i,P_{i+m}\},\label{eq:HN} \\
     H_{n,k}^2 &= H_{n,k}+\sum_{m=1}^{n-1}\sum_{i=k}^{k+n-m-1}\{P_i,P_{i+m}\}.\label{eq:Hnk}
\end{align}

Notice that $n-m = \big|\big\{k \, |\, i,i+m\in \{k, k+1, \ldots, k+n-1\}\big\}\big|$ for any $1\leq i \leq N$ and $m<n$. As a consequence, summing \eqref{eq:Hnk} over all possible $k$, and applying \eqref{Hnk_properties} produces
\begin{align}
    \sum_{k=1}^N H_{n,k}^2 & = nH_N + \sum_{m=1}^{n-1}\left(\sum_{k=1}^N\sum_{i=k}^{k+n-m-1}\{P_i,P_{i+m}\}\right) = nH_N+\sum_{m=1}^{n-1}(n-m)\sum_{i=1}^{N}\{P_i,P_{i+m}\} \label{summed_Hnk}.
\end{align}
Since $\{P_i, P_{i+m}\}\geq 0$ whenever $m\geq 2$, combining \eqref{summed_Hnk} with \eqref{eq:HN} and using that $n\leq \lfloor N/2 \rfloor$ give
\[
\sum_{k=1}^NH_{n,k}^2 \leq nH_N+(n-1)\sum_{m=1}^{\lfloor N/2\rfloor}\sum_{i=1}^N\{P_i,P_{i+m}\} = (n-1)H_N^2+H_N.
\]

On the other hand, \eqref{Hnk_properties} implies
\[
\sum_{k=1}^NH_{n,k}^2 \geq \min_{1\leq k \leq N}\gap(H_{n,k})\sum_{k=1}^NH_{n,k} = n\left(\min_{1\leq k \leq N}\gap(H_{n,k})\right) H_N,
\]
which when coupled with the previous operator inequality yields the desired bound:
\[
(n-1)H_N^2 \geq n\left(\min_{1\leq k \leq N}\gap(H_{n,k})-\frac{1}{n}\right)H_N.
\]
\end{proof}

The previous statement held for any finite sequence of positive operators on a finite-dimensional Hilbert space defined as in \eqref{eq:commutator}-\eqref{eq:fcs_hams}. We now turn to results that produce a lower bound on the spectral gap of a quantum spin Hamiltonian that is restricted to an invariant subspace. However, we first set some notation that will be used throughout the rest of the paper.

Let $\Gamma$ denote a lattice and $\cH_\Lambda = \bigotimes_{x\in\Lambda}\bC^{n_x}$ be the Hilbert space associated with a quantum spin model on any finite subset $\Lambda \subseteq \Gamma$. All existing spectral gap methods require that the model is frustration-free, meaning that the ground states of any local Hamiltonian simultaneously minimizes the energy of all of its interaction terms. After shifting each interaction term, $h_X$, by its ground state energy, i.e. $h_X-\min\spec(h_X)\geq 0$, a model is \emph{frustration-free} if and only if the kernel of every local Hamiltonian is nontrivial. For frustration-free models, let $G_\Lambda^{\#}\in\caB(\caH_\Lambda)$ be the the orthogonal projection onto the ground state space
\begin{equation}
    \caG_\Lambda^{\#}=\ker(H_\Lambda^{\#})\subseteq \cH_\Lambda, \quad \#\in\{{\rm per},\, \cdot \,\}\,.
\end{equation}
Recall that for any $\Lambda'\subseteq \Lambda$ there is a natural embedding $\cB(\cH_{\Lambda'})\ni A \mapsto A\otimes \idty_{\Lambda\setminus\Lambda'}\in \cB(\cH_\Lambda)$. We slightly abuse notation and let $G_{\Lambda'}\in\caB(\caH_\Lambda)$ denote the orthogonal projection onto $\cG_{\Lambda'}\otimes \cH_{\Lambda\setminus \Lambda'}$ for all $\Lambda'\subseteq\Lambda.$ Since the ground states of frustration-free models are the common ground states of all interaction terms, it follows that for all $\Lambda''\subseteq \Lambda'\subseteq \Lambda$:
\begin{equation}\label{eq:ff_projections}
  G_{\Lambda'}G_{\Lambda''} = G_{\Lambda''}G_{\Lambda'} = G_{\Lambda'} \quad \text{and}\quad G_{\Lambda'}G_{\Lambda}^{\rm per} = G_{\Lambda}^{\rm per}G_{\Lambda'} = G_{\Lambda}^{\rm per}\, .  
\end{equation}
Finally, for $\#\in\{{\rm per},\, \cdot \, \}$, we say that $H_\Lambda^{\#}$ is \emph{frustration-free on an invariant subspace} $\caV\subseteq \caH_\Lambda$, if it is a frustration-free Hamiltonian and
\begin{equation}\label{eq:FFonV}
    H_\Lambda^{\#}\caV \subseteq \caV, \qquad \caV \cap \caG_\Lambda^{\#} \neq \{0\}.
\end{equation}

The restriction of any operator $A\in\cB(\cH_\Lambda)$ to a subspace $\caV\subseteq \caH_\Lambda$ produces a linear map $A\restriction_{\caV}\in\cB(\caV,\caH_\Lambda)$, and we denote by $\|A\|_{\caV}$ the norm of the restriction:
\begin{equation}\label{eq:restricted_norm}
    \|A\|_{\caV} = \sup_{0\neq \psi \in \caV_\Lambda}\frac{\|A\psi\|}{\|\psi\|} \quad \text{for any} \quad A\in \caB(\caH_\Lambda),
\end{equation}
and set $\|A\|_{\caV}=0$ if $\caV=\{0\}$. When $A$ is invariant under $\caV$, the restriction is equal to
\begin{equation}\label{eq:invariant_op}
A\restriction_{\caV} = P_\caV AP_{\caV} = AP_{\caV}= P_\caV A    
\end{equation}
where $P_{\caV}\in \cB(\caH_\Lambda)$ is the orthogonal projection onto $\caV$. In this case, the operator can be block-diagonalized as $A=A\restriction_{\caV}\oplus A\restriction_{\caV^\perp}$ and $\spec(A) = \spec(A\restriction_{\caV})\cup \spec(A\restriction_{\caV^\perp})$. Moreover, if $A\geq 0$ and $\ker(A\restriction_\caV)\neq\{0\}$, the spectral gap of $A\restriction_\caV$ is its the smallest positive eigenvalue in $\caV$:
\begin{equation}\label{eq:subspace_gap}
\gap(A\restriction_{\caV}):= \inf_{\substack{0\neq \psi \in \caV \\ \psi \in \ker(A\restriction_\caV)^\perp}}\frac{\braket{\psi}{A\psi}}{\|\psi\|^2}> 0\,
\end{equation}
where by convention $\gap(A\restriction_{\caV}) = \infty$ if $\caV\subseteq \ker(A).$

We now show how the generalized Knabe criterion, Theorem~\ref{thm:Knabe}, can be combined with a coarse-graining procedure to produce a lower bound on the spectral gap of a local Hamiltonian for a finite-range, frustration-free quantum spin chain. While finite size criteria exist for multi-dimensional lattices \cite{Gosset2016,Lemm2019,Lemm2020}, we focus on quantum spin chains as this is the case of interest for Theorem~\ref{thm:main_result}.

\begin{cor}\label{cor:fsc}
Let $\Lambda\subseteq \bZ$ be a finite interval with the ring geometry, and assume that $H_\Lambda^{\rm per}$ is frustration-free on an invariant subspace $\caV_\Lambda\subseteq\caH_\Lambda$. Moreover, suppose $\Lambda = \cup_{i=1}^N X_i$ where the subintervals satisfy the following for all $1\leq i,j\leq N$ and some $1< n\leq \lfloor N/2\rfloor$
\begin{enumerate}
    \item $\caV_\Lambda$ is invariant under each $H_{X_i}$ and $H_{\Lambda_{n,k}}$ where $\Lambda_{n,k} = \bigcup_{i=k}^{n+k-1}X_i.$
    \item If $i<j$ and $X_i\cap X_j \neq \emptyset$ then $j=i+1$, or $j=N$ and $i=1$.
    \item Every interaction term of $H_\Lambda^{\rm per}$ is supported on at least one $X_i$.
\end{enumerate}
Let $\gamma = \inf_{1\leq i \leq N}\gap(H_{X_i}\restriction_{\caV_\Lambda}),$ and  $C = \sup_{1\leq i \leq N} \|H_{X_i}\|_{\caV_\Lambda}.$ Then,
\begin{equation} \label{coarse_gap}
    \gap(H_\Lambda^{\rm per}\restriction_{\caV_\Lambda})\geq \frac{\gamma n}{2C(n-1)}\left(\min_{1\leq k \leq N} \gap(H_{\Lambda_{n,k}}\restriction_{\caV_\Lambda})-\frac{C}{n}\right)\,.
\end{equation}
\end{cor}
 Note that the invariance assumptions trivially hold if $\caV_\Lambda =\caH_\Lambda$. Given condition (ii), $\caV_\Lambda$ will be invariant under all $H_{\Lambda_{n,k}}$ if it is invariant under $H_{X_{i}\cap X_{i+1}}$ for all $1\leq i\leq N$, where $N+1\equiv 1.$ Thus, for a translation invariant model, if $|X_i|$ and $|X_i\cap X_{i+1}|$ can be chosen independent of $\Lambda$ and $i$, then the constants $\gamma$ and $C$ can be appropriately bounded independent of $\Lambda$, and \eqref{coarse_gap} will produce a $\Lambda$-independent lower bound on $\gap(H_\Lambda^{\rm per}\restriction_{\caV_\Lambda})$.

\begin{proof}
    Let $P_{\caV_\Lambda}$ be the orthogonal projection onto the invariant subspace $\caV_\Lambda$. By the first assumption on the $X_i$, the ground state projection $G_{X_i}$ is invariant under $\caV_\Lambda$ as it is a spectral projection of $H_{X_i}$. Thus,  $P_i := (\idty - G_{X_i})\restriction_{\caV_\Lambda}$ satisfies \eqref{eq:invariant_op}, and is nonzero as $H_\Lambda^{\rm per}$ is frustration-free on $\caV_\Lambda$ and $\caG_\Lambda^{\rm per}\subseteq \caG_{X_i}\otimes \caH_{\Lambda\setminus X_i}$, see \eqref{eq:ff_projections}.
    
    Now, define $H_N = \sum_{i=1}^N P_i$ and $H_{n,k} = \sum_{i=k}^{k+n-1}P_i$. Trivially, $\gamma P_i \leq P_{\caV_\Lambda}H_{X_i}P_{\caV_\Lambda} \leq C P_i.$ The assumptions on the intervals guarantee that every interaction term is supported on at least one and at most two of the $X_i.$ Combining these observations yields 
    \[
H_\Lambda^{\rm per}\restriction_{\caV_\Lambda}\leq\sum_{i=1}^N H_{X_i}\restriction_{\caV_\Lambda} \leq 2H_\Lambda^{\rm per}\restriction_{\caV_\Lambda}, \quad     \gamma H_N \leq\sum_{i=1}^N H_{X_i}\restriction_{\caV_\Lambda} \leq CH_N.
    \]    
The analogous calculations also hold for $H_{n,k}$ and $H_{\Lambda_{n,k}}\restriction_{\caV_\Lambda}$. As a consequence, 
\[
\frac{\gamma}{2} H_N\leq H_\Lambda^{\rm per}\restriction_{\caV_\Lambda} \leq C H_N, \quad \text{and} \quad  \frac{\gamma}{2} H_{n,k}\leq H_{\Lambda_{n,k}}\restriction_{\caV_\Lambda} \leq C H_{n,k}
\]
These bounds imply that the positive operators $H_N$ and $H_\Lambda^{\rm per}\restriction_{\caV_\Lambda}$ (respectively, $H_{n,k}$ and $H_{\Lambda_{n,k}}\restriction_{\caV_\Lambda}$) have the same kernels in $\caV_\Lambda$. Thus, since  $H_\Lambda^{\rm per}$ and $H_{\Lambda_{n,k}}$ are invariant under $\caV_\Lambda$, by \eqref{eq:subspace_gap}
\begin{equation}\label{two-gap}
    \gap(H_\Lambda^{\rm per}\restriction_{\caV_\Lambda}) \geq \frac{\gamma}{2}\gap(H_N), \quad \gap(H_{\Lambda_{n,k}}\restriction_{\caV_\Lambda}) \leq C\gap (H_{n,k}).
\end{equation}

The first assumption on the intervals $X_i$ imply that the family of projections $P_i$, $i=1,\ldots, N$, satisfies the conditions of Theorem~\ref{thm:Knabe} since if $|i-j|>1$ in the ring geometry,
\[
P_iP_j = P_{\caV_\Lambda}(\idty-G_{X_i})(\idty-G_{X_j})P_{\caV_\Lambda}=P_{\caV_\Lambda}(\idty-G_{X_j})(\idty-G_{X_i})P_{\caV_\Lambda} = P_jP_i
\]
where we use \eqref{eq:invariant_op}, and that spatially separated observables commute. Thus, \eqref{eq:Knabe_bound} holds by Theorem~\ref{thm:Knabe}, which combined with \eqref{two-gap} produces \eqref{coarse_gap}.
\end{proof}

\subsubsection{The martingale method}\label{sec:MM} We now turn our attention to the the martingale method, which uses ground state projections to effectively localize low-lying excitations. We assume the setup and notation discussed after the proof of Theorem~\ref{thm:Knabe}, and show that Theorem~\ref{thm:MM} below holds if the following assumptions are satisfied.

\begin{assumption}[Martingale Method]\label{assump:MM}
Fix a finite volume $\Lambda\subseteq \Gamma$, and let $H_\Lambda$ be a frustration-free Hamiltonian and $\caV_\Lambda\subseteq \caH_\Lambda$ a subspace so that $\caV_\Lambda\cap \caG_\Lambda \neq \{0\}$. Assume there exists a finite sequence \[X_n \subseteq \Lambda, \quad 1\leq n \leq N,\] 
so that the local Hamiltonians and ground state projections associated with $X_n$ and $\Lambda_n := \bigcup_{i=1}^n X_i$ satisfy the following:
\begin{enumerate}
    \item (Uniform Local Gap) There exists $\gamma >0$ so that for all $1\leq n \leq N$, and all $\psi \in \caV_\Lambda$
    \[\gamma \braket{\psi}{(\idty-G_{X_n})\psi} \leq \braket{\psi}{H_{X_n}\psi}\,.\]
    \item (Absorbs the Interaction) There exists $d\geq 1$ so that $H_\Lambda \leq \sum_{i=1}^N H_{X_i} \leq d H_\Lambda.$
    \item (Sufficiently Local Sequence) There exists $\ell\geq 0$ so that for any $m\leq n\leq N$, 
    \[X_m \cap X_n \neq \emptyset \implies  m\in [n-\ell+1, n].\]
    \item (Approximate Local Excitations) Define $G_{\Lambda_0}=\idty$, $G_{\Lambda_{N+1}}=0$, and $ E_n:=G_{\Lambda_n}-G_{\Lambda_{n+1}}$ for all $0 \leq n \leq N$. Then
    \begin{equation}
        \label{eq:epsilon}
        \epsilon := \sup_{0 \leq n \leq N-1} \|G_{X_{n+1}}\|_{E_n\caV_\Lambda} < \frac{1}{\sqrt{\ell}}.
    \end{equation}
\end{enumerate}
\end{assumption}
Recall that the norm in \eqref{eq:restricted_norm} is as in \eqref{eq:epsilon}. If $\caV_\Lambda$ is invariant under $E_n$ for all $1\leq n \leq N$, then
\begin{equation}\label{eq:invariant_relation}
    \|G_{X_n}\|_{E_n\caV_\Lambda}= \|G_{X_n}E_n\|_{\caV_\Lambda}.
\end{equation}
Since both norms are zero if $E_n\caV_\Lambda = \{0\}$, the first norm always bounds the latter as
\[
 \|G_{X_n}E_n\|_{\caV_\Lambda}^2 = \sup_{0\neq \psi \in \caV_\Lambda} \frac{\|G_{X_n}E_n\psi\|^2}{\|E_n\psi\|^2+\|(\idty-E_n)\psi\|^2} \leq \sup_{0\neq \psi \in \caV_\Lambda} \frac{\|G_{X_n}E_n\psi\|^2}{\|E_n\psi\|^2} = \|G_{X_n}\|_{E_n\caV_\Lambda}^2. 
\]
The opposite inequality holds in the case that $\caV_\Lambda$ is invariant under $E_n$ since then $E_n\caV_\Lambda \subseteq \caV_\Lambda$. 

For $\caV_\Lambda=\caH_\Lambda$, \eqref{eq:invariant_relation} trivially holds and produces the original form of Assumption (iv) from \cite{Nachtergaele1996}. It will also hold for the application in this work as $\caV_\Lambda$ will be invariant under all of the interaction terms that comprise $H_\Lambda$. However, the following result holds regardless of any such invariance.

\begin{thm}[The Martingale Method \cite{Nachtergaele1996,Nachtergaele2018}]\label{thm:MM} Let $H_\Lambda$ be the local Hamiltonian associated with a frustration-free interaction. If the conditions of Assumption~\ref{assump:MM} are satisfied, then 
\begin{equation}\label{eq:gap_lb}
\braket{\psi}{H_\Lambda \psi}\geq \frac{\gamma}{d}(1-\epsilon\sqrt{\ell})^2\|\psi\|^2,\quad \forall \,\psi\in\caG_\Lambda^\perp\cap\caV_\Lambda.
\end{equation}
\end{thm}

\begin{proof} It is easy to check that the operators $E_n$, $0\leq n\leq N-1$ form a mutually orthogonal family of orthogonal projections that sum to the identity:
\begin{equation}\label{eq:E_properties}
   E_nE_m = \delta_{n,m}E_n, \qquad E_n^* = E_n, \qquad \sum_{n=0}^{N} E_n = \idty
\end{equation}
where, for first property follows by \eqref{eq:ff_projections}. Since $H_{\Lambda_N}=H_\Lambda,$ for any $\psi \in \cG_\Lambda^\perp\cap \caV_\Lambda$ it follows that
\[
\|\psi\|^2 = \sum_{n=0}^{N-1}\|E_n\psi\|^2.
\]

Consider $\|E_n\psi\|^2$ for any $0\leq n \leq N-1$, and set $n_\ell := \max\{0,n-\ell+1\}$. Using \eqref{eq:E_properties}, one finds
\begin{align}
  \|E_n\psi\|^2  = & \braket{(\idty-G_{X_{n+1}})\psi}{E_n\psi}  + \braket{\sum_{m= n_\ell}^nE_m\psi}{G_{X_{n+1}}E_n\psi} \label{eq:mm_step1}
\end{align}
where we use that Assumption~\ref{assump:MM}(iii) and the frustration-free property imply that
\[
[G_{X_{n+1}},E_m] = 0 \quad \text{if} \quad m \notin [n-\ell+1,n].
\]
Applying the bound $|\braket{\phi_1}{\phi_2}| \leq \frac{c}{2}\|\phi_1\|^2+\frac{1}{2c}\|\phi_2\|^2$ for any $c>0$ to \eqref{eq:mm_step1} then produces
\begin{align*}
    \|E_n\psi\|^2 \leq & \frac{1}{2c_1}\braket{\psi}{(\idty-G_{X_{n+1}})\psi} + \frac{c_1}{2}\|E_n\psi\|^2 +  \frac{1}{2c_2}\|G_{X_{n+1}}E_n\psi\|^2 + \frac{c_2}{2}\sum_{m=n_\ell}^n\|E_m\psi\|^2 \nonumber\\
    \leq & \frac{1}{2c_1\gamma}\braket{\psi}{H_{X_{n+1}}\psi} + \left(\frac{c_1}{2}+\frac{\epsilon^2}{2c_2}\right)\|E_n\psi\|^2+\frac{c_2}{2}\sum_{m=n_\ell}^{n}\|E_m\psi\|^2,
\end{align*}
where the last inequality uses items (i) and (iv) of Assumption~\ref{assump:MM}. Summing over $0\leq n \leq N-1$, rearranging terms, and applying Assumption~\ref{assump:MM}(ii) one arrives at
\[
2\gamma c_1\left(1-\frac{c_1}{2}-\frac{\epsilon^2}{2c_2}-\frac{c_2\ell}{2}\right)\|\psi\|^2 \leq d\braket{\psi}{H_\Lambda \psi} \quad \forall \psi\in\caG_\Lambda^\perp \cap \caV_\Lambda\,.
\]
The bound in \eqref{eq:gap_lb} is then a consequence of maximizing over $\{(c_1,c_2):c_1,c_2>0\},$ which results in choosing $c_1=1-\epsilon\sqrt{\ell}$ and $c_2=\epsilon/\sqrt{\ell}.$
\end{proof}

The following is an immediate consequence of Theorem~\ref{thm:MM} if $H_\Lambda$ is frustration-free on $\caV_\Lambda$.

\begin{cor}\label{cor:MM} Suppose that $H_\Lambda$ is a local Hamiltonian for a frustration-free interaction. Moreover, assume Assumption~\ref{assump:MM} holds for an invariant subspace $\caV_\Lambda$ of $H_\Lambda$. Then,
\begin{equation}\label{eq:MM_gap_IS}
    \gap(H_\Lambda\restriction_{\caV_\Lambda}) \geq \frac{\gamma}{d}(1-\epsilon\sqrt{\ell})^2.
\end{equation}

\end{cor}

\begin{proof}
Since $\caV_\Lambda$, $\caG_\Lambda$ and their orthogonal complements are all invariant under $H_\Lambda$, it follows that
\[
\caG_\Lambda^{\perp} \cap \caV_\Lambda = (\caG_\Lambda \cap \caV_\Lambda)^{\perp}\cap \caV_\Lambda = \ker(H_\Lambda\restriction_{\caV_\Lambda})^\perp \cap \caV_\Lambda.
\]
Hence, \eqref{eq:MM_gap_IS} follows immediately from combining \eqref{eq:subspace_gap} with \eqref{eq:gap_lb}.
\end{proof}

\section{Invariant subspaces of the truncated Haldane pseudopotential}\label{sec:subspaces}

In this section, we identify the invariant subspaces that will be used to prove Theorem~\ref{thm:main_result} via the strategy outlined in Section~\ref{sec:gap_strategy}. Invariant subspaces that support the edge states will also be briefly discuss. All of these subspaces will share the property that they are invariant under all of the interaction terms ($n_xn_{x+2}$ and $q_x^*q_x$) that comprise a Hamiltonian with either OBC or PBC, see \eqref{eq:OBC_Ham}-\eqref{eq:interaction} and \eqref{eq:PBC_Ham}, and will be spanned by a subset of the configuration basis
\begin{equation}\label{eq:conifguration_basis}
    \cB_\Lambda=\left\{\ket{\mu} : \mu = (\mu_1,\mu_2,\ldots, \mu_{|\Lambda|})\in\{0,1\}^{|\Lambda|}\right\}\subseteq \caH_\Lambda := \bigotimes_{x\in\Lambda}\bC^2.
\end{equation} 
 As many of the proofs in this section require only elementary calculations, we focus the discussion on their main ideas, and point the interested reader to \cite{Warzel:2022, Warzel:2023} for the precise details.

\subsection{Ground state tilings for the periodic model}\label{sec:vmd_tilings}
Consider a finite interval $\Lambda$ with PBC. We identify a subspace that contains $\cG_\Lambda^{\rm per}$ and is invariant under all interaction terms that comprise $H_\Lambda^{\rm per}$. Notice that the truncated Haldane pseudopotential is frustration-free, as every interaction term is non-negative and $\ket{0}^{\otimes |\Lambda|}\in \ker(H_\Lambda^{\rm per})$. As a consequence,
\begin{equation}\label{eq:ff}
    \caG_\Lambda^{\rm per} = \ker(H_\Lambda^{\rm per}) = \bigcap_{x\in \Lambda}\left(\ker(n_xn_{x+2}) \cap \ker(q_x)\right).
\end{equation}
Moreover, $\caB_\Lambda$ is an eigenbasis of every electrostatic term $n_xn_{x+2}$, and so 
\begin{equation} \label{eq:electrostatic_gs}
   \bigcap_{x\in \Lambda}\ker(n_xn_{x+2}) = \Span\left\{\ket{\mu}\in \cB_\Lambda : \mu_x\mu_{x+2}=0 \; \forall x\in\Lambda\right\}.
\end{equation}
It is left to identify a subspace of \eqref{eq:electrostatic_gs} that is invariant under all $q_x^*q_x$, and contains the joint kernel from \eqref{eq:ff}.
 
 The restriction of any $\ket{\mu}$ from the spanning set of \eqref{eq:electrostatic_gs} to an interval of four sites has at most two particles, and it is easy to check that $ \ket{\mu}\in\ker(q_x)$ in all cases except $\ket{\mu}\restriction_{[x-1,x+2]} \in \{\ket{1001}, \, \ket{0110}\}.$
One can also verify that $\Span\{\ket{1001}, \, \ket{0110}\}$ is invariant under $q^*q\in\cB((\bC^2)^{\otimes 4)}$, and
\begin{equation}\label{eq:hopping_gs}
 \ket{\mu_L}\otimes\left(\ket{1001}+\lambda\ket{0110}\right)\otimes\ket{\mu_R}  \in \ker(q_x), 
\end{equation}
for any $\mu_L\in \{0,1\}^{x-2}$ and $\mu_R\in \{0,1\}^{|\Lambda|-x-2}$. However, even if $\ket{\mu}$ belongs to \eqref{eq:electrostatic_gs}, it is not true that any $\ket{\mu'}$ obtained from replacing one or more subconfigurations $0110$ with $1001$ also belongs to this set. For example, $\ket{\ldots011001\ldots}$ can belong to \eqref{eq:electrostatic_gs}, but $\ket{\ldots100101\ldots}$ does not. 

We restrict the spanning set from \eqref{eq:electrostatic_gs} to states that remain in \eqref{eq:electrostatic_gs} under all such replacements. These states can be characterized by domino tilings of $\Lambda.$ There are three types of tiles:
\begin{enumerate}
    \item The void, $V=(0)$, covers one site and contains no particles.
    \item The monomer, $M=(100)$, covers three sites and has a single particle on the first site.
    \item The dimer, $D=(011000)$, covers six sites and has particles on the second and third sites.
\end{enumerate}

Any tiling $T$ that covers all sites of the ring $\Lambda$ with these tiles is called a \emph{Void-Monomer-Dimer (VMD) tiling}, see Figure~\ref{fig:VMD_tiling}. The set of all tilings is denoted $\caT_\Lambda^{\rm per}$, and the map
\[
\sigma_\Lambda^{\rm per}:\caT_\Lambda^{\rm per} \to \{0,1\}^{|\Lambda|}, \quad T\mapsto \sigma_\Lambda^{\rm per}(T),
\]
 where $\sigma_\Lambda^{\rm per}(T)$ is the configuration whose occupation matches that of the tiling, is injective. This follows from noticing that any configuration $\ket{\mu}\in\ran(\sigma_\Lambda^{\rm per})$ with neighboring occupied sites must be covered by a dimer, and any isolated occupied site must be covered by a monomer. After first placing all dimers and monomers, all other sites are vacant and, hence, must be covered by a void. This constructs the unique tiling corresponding to $\mu$. To simplify notation, we will use $\ket{T}$ to denote the occupation state $\ket{\sigma_\Lambda^{\rm per}(T)}.$

The range of $\sigma_\Lambda^{\rm per}$ can also be precisely identified, and will be useful for bounding $E_0(\caV_\Lambda^\perp)$ for the invariant subspace strategy.
\begin{prop}\label{prop:tiling_configs}
A configuration $\mu\in\{0,1\}^{|\Lambda|}$ belongs to $\ran(\sigma_\Lambda^{\rm per})$ if and only if the following hold for all $x\in\Lambda$ where $x\equiv x+|\Lambda|:$
\begin{enumerate}
    \item If $\mu_x=1$ and $\mu_{x\pm 1}=0$, then $\mu_{x\pm 2}=0$.
    \item If $\mu_{x}=\mu_{x+1}=1$, then the following hold where $s=x+1/2$:
    \begin{enumerate}
        \item The first three sites on either side of this pair are vacant: $\mu_{s\pm 3/2}=\mu_{s\pm 5/2}=\mu_{s\pm 7/2}=0$.
        \item The next two sites have at most one particle: $\mu_{s+ 9/2}\mu_{s+ 11/2}=0$ and $\mu_{s-9/2}\mu_{s-11/2}=0$. 
    \end{enumerate}
\end{enumerate}
\end{prop}
\begin{proof}
    That any tiling configuration $\sigma_\Lambda^{\rm per}(T)$ satisfies these conditions is easy to verify from considering the occupied sites that result from laying any combination of three tiles next to each other. The reverse direction is proved using the same argument that shows $\sigma_\Lambda^{\rm per}$ is injective. The only additional observation needed is that the occupation constraints from items (i)-(ii) guarantee there are enough vacant sites so that tiles covering occupied sites do not overlap one another.
\end{proof}

\begin{figure}
    \centering
    \includegraphics[scale=.3]{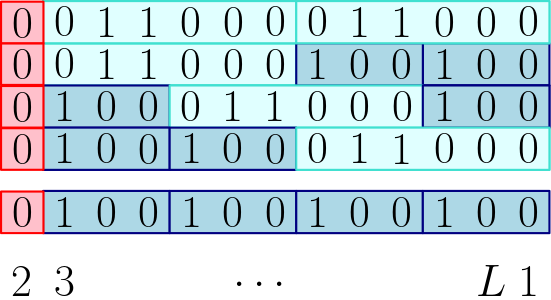}
    \caption{An equivalence class of VMD tilings generated by the replacement rule. We use the void on the second site to unravel the periodic tiling to an interval.}
    \label{fig:VMD_tiling}
\end{figure}
The subspace of all VMD tilings
\begin{equation}
\caC_\Lambda^{\rm per} := \Span\{\ket{T}: T\in \cT_\Lambda^{\rm per}\}
\end{equation}
is the invariant subspace $\caV_\Lambda=\caC_\Lambda^{\rm per}$ we use to prove Theorem~\ref{thm:main_result} via the bulk gap strategy from Section~\ref{sec:gap_strategy}. To see that this contains $\caG_\Lambda^{\rm per}$, we first use \eqref{eq:hopping_gs} to partition $\caC_\Lambda^{\rm per}$ into mutually orthogonal invariant subspaces. The bidirectional tile replacement rule
\[
(100)(100) \leftrightarrow (011000) \quad \text{or, equivalently,} \quad MM\leftrightarrow D
\]
generates an equivalance relation on $\caT_\Lambda^{\rm per}$, where $T\leftrightarrow T'$ if and only if the tiling $T'$ can be constructed from $T$ after a finite number of tile replacements. Any tiling consisting of only voids and monomers is called a \emph{root tiling}. By replacing all dimers in $T$ with a pair of monomers, it is clear that there is a unique root tiling $R$ such that $T\leftrightarrow R$. Hence, 
\[\caR_\Lambda^{\rm per}=\{R\in \cT_\Lambda^{\rm per}: R \text{ contains no dimers } D\}\]
is a set of representatives for the equivalence classes:
\[
\caT_\Lambda^{\rm per} = \biguplus_{R\in \caR_\Lambda^{\rm per}} \caT_\Lambda^{\rm per}(R), \quad \caT_\Lambda^{\rm per}(R) = \{T\in \caT_\Lambda^{\rm per}: T\leftrightarrow R\},
\]
see Figure~\ref{fig:VMD_tiling}. This combined with the injectivity of $\sigma_\Lambda^{\rm per}$ and orthogonality of $\caB_\Lambda$ imply that  \begin{equation}\label{eq:tiling_properties}
    \caC_\Lambda^{\rm per} = \bigoplus_{R\in \caR_\Lambda^{\rm per}} \caC_\Lambda^{\rm per}(R)\quad\text{and}\quad\caC_\Lambda^{\rm per}(R)\perp\caC_\Lambda^{\rm per}(R') \quad \forall \, R\neq R'
\end{equation}
where $\caC_\Lambda^{\rm per}(R): = \Span\left\{\ket{T}: T\leftrightarrow R\right\}$ is the \emph{VMD space generated by $R$}. The following result shows that each $\caC_\Lambda^{\rm per}(R)$ supports a unique ground state, and that these form an orthogonal basis for $\caG_\Lambda^{\rm per}$.

\begin{lemma}\label{lem:gss}
   Fix a ring $\Lambda=[1,L]$ with $L\geq 6$. Then $\caC_\Lambda^{\rm per}(R)$ is an invariant under all interaction terms $\{n_xn_{x+2}, \, q_x^*q_x$: $x\in\Lambda\}$ for any $R\in \caR_\Lambda^{\rm per}$. Moreover,
    \begin{equation}
       \caG_\Lambda^{\rm per} = \Span\{\psi_\Lambda^{\rm per}(R) : R\in\caR_\Lambda^{\rm per}\}
   \end{equation}
 where, denoting by $n_D(T)$ the number of dimers in the tiling $T$, $\psi_\Lambda^{\rm per}(R)$ is the VMD state
   \begin{equation}
       \psi_\Lambda^{\rm per}(R) =\sum_{T\leftrightarrow R}\lambda^{n_D(T)}\ket{T}\in\caC_\Lambda^{\rm per}(R).
   \end{equation} 
\end{lemma}
We omit the details of this proof as they are straightforward, and can be found in \cite{Warzel:2022}. The idea for showing that $\caC_\Lambda^{\rm per}(R)$ is invariant under all interaction terms and supports the unique ground state $\psi_\Lambda^{\rm per}(R)$ is a consequence of the observations made in \eqref{eq:ff}-\eqref{eq:hopping_gs} and the surrounding discussion. That these states form an orthogonal basis for $\caG_\Lambda^{\rm per}$ results from combining \eqref{eq:tiling_properties} with Proposition~\ref{prop:tiling_configs} to prove that $\psi(\mu')\neq 0$ only if $\mu' \in \ran (\sigma_{\Lambda}^{\rm per})$ for any $\psi=\sum_{\mu\in\{0,1\}^{|\Lambda|}}\psi(\mu)\ket{\mu}\in \cG_\Lambda^{\rm per}$. 

An immediate consequence of Lemma~\ref{lem:gss} is that $\caC_\Lambda^{\rm per}(R)$ is invariant under $H_\Lambda^{\rm per}.$ 

\subsection{Ground state tilings for the open boundary model}\label{sec:obc}
As the martingale method will be used to prove the threshold criterion of the finite size criterion, we also need to know the ground states of $H_\Lambda$. If $\Lambda=[1,L]$, the frustration-free property implies
\[
\caG_\Lambda =\bigcap_{x=1}^{L-2}\ker(n_xn_{x+2)} \cap \bigcap_{x=2}^{L-2} \ker(q_x).
\]
A similar tiling description as in Section~\ref{sec:vmd_tilings} holds for the model with OBC. However, several additional \emph{boundary tiles} are needed to complete the description. As one can always embed an open interval as a subinterval of a ring, it is not surprising that these boundary tiles are obtained from restricting a VMD tiling to an interval. Discarding any truncated tiles that can be built from the void and monomer tiles, this produces the following sets of boundary tiles:

\emph{The left boundary tiles:}
\begin{enumerate}
    \item A dimer $B_l = (11000)$ covering five sites with particles on the first and second site.
\end{enumerate}

\emph{The right boundary tiles:}
\begin{enumerate}
    \item A monomer $M_1 = (1)$ covering one site with a particle.
    \item A monomer $M_2=(10)$ covering two sites with a particle on the first site.
    \item A dimer $B_r = (011)$ covering three sites with particles on the second and third sites.
    \item A dimer $D_1 = (0110)$ covering four sites with particles on the second and third sites.
    \item A dimer $D_2 = (01100)$ covering five sites with particles on the second and third sites.
\end{enumerate}
With these additional tiles, a \emph{boundary-void-monomer-dimer (BVMD) tiling} $T$ of an interval $\Lambda$ is any covering of $\Lambda$ with the bulk tiles $\{V, \, M, \, D\}$ from the previous section or the above boundary tiles. As suggested by the notation, if used, a boundary tile can only be placed on its respective boundary in the interval, see Figure~\ref{fig:BVMD_tiling}.

Denoting by $\caT_\Lambda$ the set of all BVMD tilings, there is again an injective map $\sigma_\Lambda: \cT_\Lambda \to \{0,1\}^{|\Lambda|}$ which identifies any tiling with the configuration that agrees with the tiling. Every VMD tiling can be viewed as a BVMD tiling by cutting the VMD tiling between the first and last site and appropriately adjusting the boundary tiles. Thus, one can identify $\caT_\Lambda^{\rm per}\subset \caT_\Lambda$ and $\sigma_\Lambda\restriction_{\caT_\Lambda^{\rm per}} = \sigma_\Lambda^{\rm per}$. Thus, we again write $\ket{T}\equiv\ket{\sigma_\Lambda(T)}$. The analog of Proposition~\ref{prop:tiling_configs} holds for $\ran(\sigma_\Lambda)$, with the only modification being that the identification $x\equiv x+|\Lambda|$ is replaced with the convention that all sites in $\bZ\setminus \Lambda$ are considered to be vacant.

\begin{figure}
    \centering
    \includegraphics[scale=.3]{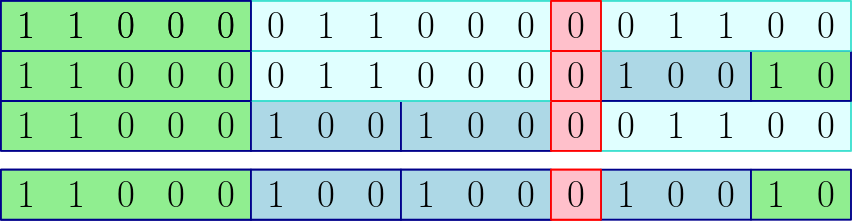}
    \caption{A BMVD tiling equivalence class generated by the replacement rules. Boundary tiles can only be placed on their designated boundary, but are not required to be used.}
    \label{fig:BVMD_tiling}
\end{figure}
Replacement rules are again used to partition the BVMD tiling space 
\[\caC_\Lambda = \Span\{\ket{T} : T\in \caT_\Lambda\}\]
into subspaces that each support a unique ground state. Introducing the alternate notation $M_3\equiv M$ and $D_3\equiv D$, the three bidirectional BVMD replacement rules are
\begin{equation}\label{eq:replacements}
  MM_i \leftrightarrow D_i, \quad i=1,2,3\,. 
\end{equation}
(We note that the boundary dimers $B_l$ and $B_r$ are not subject to replacement rules as any change in particle content from such a replacement would violate either the particle number or center of mass symmetry of $H_\Lambda$.) The replacement rules again generate an equivalence relation on $\caT_\Lambda$ where $T\leftrightarrow T'$ if a finite number of BVMD tile replacements can be used to transform $T$ into $T'$. The set of all BVMD root tilings,
\[
\caR_\Lambda = \left\{R\in \caT_\Lambda : R \text{ does not contain any dimers } D_i,\; i =1,2,3\right\},
\]
is again a set of representatives for this equivalence relation, see Figure~\ref{fig:BVMD_tiling}. Let $\caC_\Lambda(R) = \Span\{\ket{T}: T\leftrightarrow R\}$ denote the BVMD space generated by $R\in\caR_\Lambda$. Then, $\caC_\Lambda$ can be decomposed as
\[
\caC_\Lambda = \bigoplus_{R\in\caR_\Lambda} \caC_\Lambda(R) \quad \text{where} \quad \caC_\Lambda(R) \perp \caC_{\Lambda}(R') \; \text{whenever}\; R\neq R'.
\]
It is not surprising that the analog of Lemma~\ref{lem:gss} holds for the model with OBC.

\begin{lemma}\label{lem:obc_gss} Let $\Lambda=[1,L]$ with $L\geq 6$. For any $R\in \caR_\Lambda$, the space $\caC_\Lambda(R)$ is invariant under all interaction terms $n_xn_{x+2}$ or $q_x^*q_x$ that comprise $H_\Lambda$, see \eqref{eq:OBC_Ham}. Moreover, if $L\neq 7$, then 
\begin{equation}
    \caG_\Lambda = \Span\{\psi_\Lambda(R): R\in\caR_\Lambda\}
\end{equation}
where, denoting by $n_D(T)$ the number of dimers (bulk or boundary) contained in the dimer $T$, $\psi_\Lambda(R)$ is the BVMD state
\begin{equation}
\psi_{\Lambda}(R) = \sum_{T\leftrightarrow R} \lambda^{n_D(T)} \ket{T}\in \caC_\Lambda(R).
\end{equation}
\end{lemma}

 The proof here follows the analogous steps as those used to prove Lemma~\ref{lem:gss}. Moreover, as proved in \cite{Nachtergaele2021}, when $|\Lambda|=7$ an additional anomalous ground state appears, and an orthogonal basis for the ground state space is $\{\psi_\Lambda(R) : R\in\caR_\Lambda\} \cup \left\{\ket{1100011}\right\}.$

\subsection{Edge state tilings}\label{sec:edge_states}
We briefly introduce the invariant subspaces that support the edge states of $H_\Lambda$, which can also be described by domino tilings. Edge tilings agree with BVMD tilings on the interior of $\Lambda$, but violate one of the occupation conditions of Proposition~\ref{prop:tiling_configs} near the boundary of $\Lambda$. As such, they belong to $(\caC_\Lambda^{\rm per})^\perp$.

Several additional boundary tiles are needed to define edge tilings. Two of these tiles are used for the edge root tilings, namely
\begin{enumerate}
    \item The left edge root tile: $E_l = (1100100)$
    \item The right edge root tile: $E_r = (10011).$
\end{enumerate}
These result from removing the inner most zero of $B_l$, respectively $B_r$, and then appending a monomer $M$ to the interior. This breaks condition (iia) of Proposition~\ref{prop:tiling_configs}. Four other edge tiles and replacement rules are need to ensure these subspaces are invariant under all interaction terms of $H_\Lambda$. Each of these tiles also breaks one of the criteria from Proposition~\ref{prop:tiling_configs}.
\begin{enumerate}
    \item The additional left edge tiles are $T_l^1=(1011000)$ and $T_l^2=(1100011000)$. These satisfy the replacement rules:
    \begin{equation}\label{eq:replacement_el}
            E_l \leftrightarrow T_l^1, \quad E_lM \leftrightarrow T_l^2.
    \end{equation}
    \item The additional right edge tiles are $T_r^1=(01101)$ and $T_r^2=(01100011)$. These satisfy the replacement rules:
    \begin{equation}\label{eq:replacement_er}
            E_r \leftrightarrow T_r^1, \quad ME_r \leftrightarrow T_r^2.
    \end{equation}
\end{enumerate}

\begin{figure}
    \centering
    \includegraphics[scale=.3]{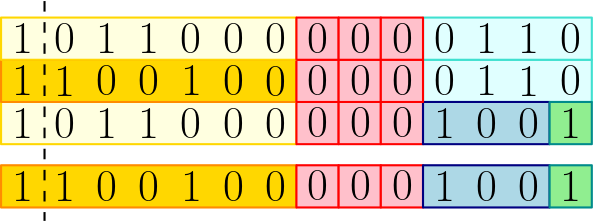}
    \caption{An edge tiling equivalence class generated by the replacement rules. Note that removing the boundary site of the edge tile produces a BVMD tiling.}
    \label{fig:Edge_tiling}
\end{figure}

Edge tiles can only be placed on their indicated boundary. An \emph{edge tiling} is any tiling that has \emph{at least one} edge tile, and otherwise follows the rules of a BVMD tiling, see Figure~\ref{fig:Edge_tiling}. It is an \emph{edge root tiling} if it contains no $D_i$, $i=1,2,3$ and no $T_l^j,T_r^j,$ $j=1,2$. Combining the replacement rules \eqref{eq:replacements} with \eqref{eq:replacement_el}-\eqref{eq:replacement_er} forms an equivalence relation on the set of all edge tilings, $\caT_\Lambda^e$, where two edge tilings are equivalent if one is transformed to the other after a finite number of tile replacements. The equivalence classes are indexed by the set of edge root tilings, $\caR_\Lambda^e$, and it is easy to check that
\[
\caC_\Lambda^e(R) = \{\ket{T} : T\leftrightarrow R\}, \qquad R\in \caR_\Lambda^e
\]
is invariant under each interaction term of $H_\Lambda.$  We continue the edge state discussion in Section~\ref{sec:edge_states}.

\subsection{Tiling space framentation and isospectral properties}
 The (B)VMD spaces are designed to be invariant under all of interaction terms of a local Hamiltonians, and so they remain invariant subspaces if some of those interaction terms are removed. As such, for any subinterval $\Lambda'\subseteq \Lambda$,
 \[
 H_{\Lambda'}\caC_\Lambda^{\#}(R) \subseteq \caC_\Lambda^{\#}(R), \quad \#\in \{{\rm per}, \, \cdot \,\} .
 \]
 Since removing interaction terms lifts constraints on the invariance, $\caC_\Lambda^\#(R)$ decomposes into multiple invariant subspaces of $H_{\Lambda'}$. Namely, this becomes a sum of spaces where the replacement relations only apply to $\Lambda'.$ We analyze this fragementation as it will be helpful for applying the spectral gap methods from Section~\ref{sec:gap_methods} to $\caV_\Lambda =\caC_\Lambda^{\rm per}.$ 

The decomposition of $\caC_\Lambda^\#(R)$ for any $R\in\caR_{\Lambda}^{\#}$ into invariant subspaces of $H_{\Lambda'}$ can be determined by considering the truncation of any $T\leftrightarrow R$ to $\Lambda'$. As the boundary tiles in Section~\ref{sec:obc} were defined to account for all possible truncations, for any $T\in\caT_\Lambda^{\#}(R)$, there is a unique $T'\in\caT_{\Lambda'}$ such that
\[\sigma_{\Lambda}^{\#}(T)\restriction_{\Lambda'} = \sigma_{\Lambda'}(T').\]
We call $T'$ the \emph{truncation} of $T$ to $\Lambda'$ and write $T'=T\restriction_{\Lambda'},$ see Figure~\ref{fig:truncation}.

The truncation of a root tiling $R\in\caR_\Lambda^{\#}$ to $\Lambda'\subseteq \Lambda$ is always a root tiling $R\restriction_{\Lambda'}\in\caR_{\Lambda'}$. However, there may be other $T\leftrightarrow R$ such that $T\restriction_{\Lambda'}\in \caR_{\Lambda'}$. The number of distinct $R'\in\caR_{\Lambda'}$ obtained as truncations of $T\in\caT_\Lambda^{\#}(R)$ only depends if the truncation separates the particles of two neighboring monomers in $R$. Said differently, whether either boundary of $\Lambda'$ lays in the interior of a sub-configuration $1001$ in $R$, see Figure~\ref{fig:truncation}. For $R\in\caR_{\Lambda}^{\#}$ and $\Lambda'\subseteq \Lambda$ fixed, let $n_R(\Lambda')\in\{0,1,2\}$ be the number of neighboring monomer pairs whose particles are separated by the truncation $R\restriction_{\Lambda'}$. These $n_R(\Lambda')$ pairs of neighboring monomers are precisely those for which, when they are replaced by a dimer in $R$, the particle content of both $\Lambda'$ and $\Lambda\setminus\Lambda'$ changes. Every other tile replacement in $R$ either does not change the particle content on $\Lambda'$, or is in one-to-one correspondence with a tile replacement on $\Lambda'.$ In the case of OBC, this means that given a root $R'\in\caR_{\Lambda'}$
\[
 (\mu_l, \sigma_{\Lambda'}(R'), \mu_r) \in \ran(\sigma_\Lambda\restriction_{\caT_\Lambda(R)}) \iff (\mu_l, \sigma_{\Lambda'}(T'), \mu_r) \in \ran(\sigma_\Lambda\restriction_{\caT_\Lambda(R)}) \quad \forall \; T'\leftrightarrow R'.
\]
where  $\mu_l\in\{0,1\}^{a-1}$ and $\mu_r\in\{0,1\}^{L-b}$ given that $\Lambda'=[a,b]\subseteq [1,L]= \Lambda$.

For OBC, there are $2^{n_R(\Lambda')}$ distinct root tilings of $\caR_{\Lambda'}$ produced from truncating any $T\leftrightarrow R$ to $\Lambda'$. These are  obtained from the $2^{n_R(\Lambda')}$ ways of replacing or not replacing the $n_R(\Lambda')$ pairs of neighboring monomers whose particles are separated by the truncation with dimers in $R$, and then truncating the resulting tiling to $\Lambda'$, see Figure~\ref{fig:truncation}.
 Enumerating these root tilings as $R_k'\in \caR_{\Lambda'}$, $1\leq k \leq 2^{n_R(\Lambda')}$, the BVMD space generated by $R$ decomposes as
\begin{equation}\label{eq:OBC_fragmentation}
\caC_\Lambda(R) = \bigoplus_{k=1}^{2^{n_R(\Lambda')}} \bigoplus_{\substack{T\leftrightarrow R : \\ T\restriction_{\Lambda'}=R_k'}} \ket{T\restriction_{\Lambda_l'}} \otimes \caC_{\Lambda'}(R_k') \otimes \ket{T\restriction_{\Lambda_r'}},    
\end{equation}
where $\phi\otimes \caV \otimes \psi := \{\phi\otimes \xi \otimes \psi : \xi \in \caV\}$ for any set of vectors $\caV$, and $\Lambda_l'=[1,a-1]$, $\Lambda_2'=[b+1,L]$. By convention, we set $\ket{T\restriction_{\emptyset}}=1$.

\begin{figure}
    \centering
    \includegraphics[scale=.3]{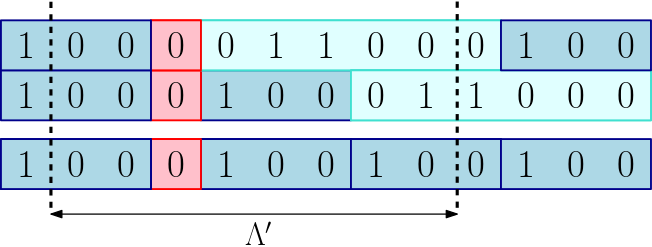}
    \caption{A truncation of a BVMD equivalence class. The right boundary of $\Lambda'$ separates the particles of a pair of neighboring monomers, but the left boundary does not. Hence, $n_R(\Lambda')=1$. The bottom two tilings produces the two distinct root tilings on $\Lambda'$ from obtained from truncating any $T\leftrightarrow R.$}
    \label{fig:truncation}
\end{figure}

The same number of distinct truncated root tilings of $\Lambda'$, is also produced in the case that $\Lambda$ has PBC as long as $|\Lambda|\geq |\Lambda'|+4$. This ensures that if both boundaries of $\Lambda'$ separate the particles of a pair of neighboring monomers, then both pairs can be simultaneously replaced by \emph{different} dimers. Under this constraint, the same analysis applies and
\begin{equation}\label{eq:PBC_fragmentation}
\caC_\Lambda^{\rm per}(R) = \bigoplus_{k=1}^{2^{n_R(\Lambda')}} \bigoplus_{\substack{T\leftrightarrow R : \\ T\restriction_{\Lambda'}=R_k'}} \caC_{\Lambda'}(R_k) \otimes \ket{T\restriction_{\Lambda\setminus \Lambda'}}.    
\end{equation}
The following is then an immediate consequence of \eqref{eq:OBC_fragmentation}-\eqref{eq:PBC_fragmentation}, and Lemmas \ref{lem:gss}-\ref{lem:obc_gss}.

\begin{prop}\label{prop:counting_gs}
    Suppose that $\Lambda'\subseteq \Lambda$ are two intervals with OBC. Then for any $R\in \caR_\Lambda$,
    \begin{equation}\label{eq:counting_gs}
        \dim\left(\caC_\Lambda(R) \cap (\caG_{\Lambda'}\otimes \cH_{\Lambda\setminus \Lambda'})\right) = \sum_{k=1}^{2^{n_R(\Lambda')}}\sum_{\substack{T\leftrightarrow R : \\ T\restriction_{\Lambda'}=R_k'}} 1
    \end{equation}
    This also holds in the case that $\Lambda$ has PBC as long as $|\Lambda|\geq |\Lambda'|+4.$
\end{prop}

We use this to determine a basis for the subspace $\caG_\Lambda^\perp\cap (\caG_{\Lambda'}\otimes \caH_{\Lambda\setminus \Lambda'})$ of interest in the next section. Any BVMD tiling $T$ can be written as an order tiling $T=(T_1, T_2, \ldots, T_k)$ where $T_1$ covers the first site, $T_2$ is the tile to the right of $T_1$, and so forth. Similarly, $T$ can be partitioned into subtilings for any $1\leq \ell <k$ as $T=(T',T'')$ where $T'=(T_1, \ldots, T_\ell)$ and $T''=(T_{\ell+1}, \ldots, T_k)$. Since the BVMD replacement rules only apply to the tiles $\{M_i, \,  D_i : i=1,2,3\}$, the particle content of any non-monomer tiles in a root $R$ does not change for any $T\leftrightarrow R$. As a consequence, $\psi_\Lambda(R)$ factorizes over all non-monomer tiles in $R$. Thus, we can partition any root tiling $R=(\tilde{R},M_n^{(i)})\in \caR_\Lambda$ into subtilings where $\tilde{R}$ does not end in a monomer and $M_n^{(i)}:=(M,\ldots, M,M_i)\in \caR_{[1,3(n-1)+i]}$ has $n\geq 0$ consecutive monomers, the last of which has length $i$. Then
\begin{equation}\label{eq:factored_bvmd}
    \psi_\Lambda(R) = \psi_{\tilde{\Lambda}_n^{(i)}}(\tilde{R})\otimes \varphi_n^{(i)},
\end{equation}
where $\tilde{\Lambda}_n^{(i)}:=\Lambda\setminus\Lambda_n^{(i)}$, $\Lambda_n^{(i)}$ is the last $3(n-1)+i$ sites of $\Lambda$,
\begin{equation}
    \label{eq:TTstate}
    \varphi_n^{(i)} = \psi_{[1,3(n-1)+i]}(M_n^{(i)}), \qquad n\geq 1
\end{equation}
and $\varphi_0^{(i)}=1$ for all $i.$

 It is easy to see from considering the set of tilings $T\leftrightarrow M_n^{(i)}$ that $\varphi_n^{(i)}$ satisfies the following relations for all for all $1\leq i\leq 3$, where $\varphi_n := \varphi_n^{(3)}$,
\begin{align}
    \varphi_n^{(i)} & =  \varphi_n^{(j)}\otimes \ket{0}^{i-j} & \forall n\geq 1, \, i\geq j \geq 1 \label{eq:phi_1}\\
    \varphi_{l+r}^{(i)} & =  \varphi_l\otimes \varphi_r^{(i)}+ \lambda \varphi_{l-1}\otimes \ket{D}\otimes \varphi_{r-1}^{(i)} & \forall r\geq 2, \, l\geq 1\label{eq:phi_2}\\
     \varphi_n^{(i)} & = \varphi_{n-1}\otimes\ket{M_i} + \lambda \varphi_{n-2}\otimes \ket{D_i} & \forall n\geq 2.\label{eq:phi_3}
\end{align}

The particle content of the last $i$ sites of $D_i$ is different from that of the monomer $M_i$ for all $i$, and so by the orthonormality of the configuration basis combined with \eqref{eq:phi_3} imply that
$\|\varphi_n^{(i)}\|^2 = \|\varphi_{n-1}\|^2 + |\lambda|^2\|\varphi_{n-2}\|^2$.
This can be used to prove
\begin{equation} \label{eq:beta}
    \beta_n := \frac{\|\varphi_{n-1}\|^2}{\|\varphi_n\|^2}= \frac{1}{\beta_+}\frac{1-\beta^n}{1-\beta^{n+1}}
\end{equation}
where $\beta_{\pm} = \big(1\pm \sqrt{1+4|\lambda|^2}\big)/2$ and $\beta = \beta_-/\beta_+\in(-1,0),$ see \cite{Nachtergaele2021}. Given these properties, it is straightforward to verify that for any root tiling $R=(\tilde{R},M_n^{(i)})$ with $n\geq 2$ as in \eqref{eq:factored_bvmd},
\begin{equation}\label{eq:eta}
    \braket{\psi_\Lambda(R)}{\eta_\Lambda(R)} = 0, \quad \text{where} \quad \eta_\Lambda(R) = \psi_{\tilde{\Lambda}_n^{(i)}}(\tilde{R})\otimes \eta_n^{(i)} \in \caC_\Lambda(R)
\end{equation}
and 
\begin{equation}\label{eq:TT_eta}
\eta_n^{(i)} = -\overline{\lambda}\beta_{n-1}\varphi_{n-1}\otimes \ket{M_i} + \varphi_{n-2}\otimes \ket{D_i} \in \caC_{[1,3(n-1)+i]}(M_n^{(i)}).
\end{equation}
These observations along with Proposition~\ref{prop:counting_gs} culminate in the following result, which will be key for the application of the martingale method in Section~\ref{sec:gap_proof}.

\begin{lemma}\label{lem:ob_tilings}
    Suppose $\Lambda'=[1,L-3]$ and $\Lambda = [1,L]$ for some $L\geq 4.$ If $\caR_\Lambda^{MM}\subseteq \caR_\Lambda$ denotes the set of root tilings that end in two or more consecutive monomers, see \eqref{eq:factored_bvmd}, then for any $R\in\caR_{\Lambda}^{MM}$
    \begin{equation}\label{eq:RMM_basis}
    \caC_\Lambda(R) \cap  (\caG_{\Lambda'}\otimes \caH_{\Lambda\setminus \Lambda'}) =
      \Span\{\psi_\Lambda(R), \, \eta_\Lambda(R)\}
    \end{equation}
    where $\eta_\Lambda(R)$ is as in \eqref{eq:eta}. Moreover, for any $R\in \caR_\Lambda\setminus \caR_\Lambda^{MM}$,
    \begin{equation}\label{eq:non_RMM_basis}
            \caC_\Lambda(R) \cap  (\caG_{\Lambda'}\otimes \caH_{\Lambda\setminus \Lambda'}) =
      \Span\{\psi_\Lambda(R)\}.
    \end{equation}
\end{lemma}

\begin{proof}
    Fix $R\in\caR_\Lambda$. Since the left boundaries of $\Lambda$ and $\Lambda'$ agree, $R\restriction_{\Lambda'}$ does not separate the particles of any pair of neighboring monomers if $R\in \caR_\Lambda\setminus \caR_\Lambda^{MM}$, and separates the particles from one pair of neighboring monomers if $R\in \caR_\Lambda^{MM}$.

    For $R\in \caR_\Lambda\setminus \caR_\Lambda^{MM}$, the particle content on the last three sites of any $T\leftrightarrow R$ is invariant under tile replacements. Therefore, by Proposition~\ref{prop:counting_gs}
    \[
    \dim\left(\caC_\Lambda(R)\cap (\caG_{\Lambda'}\otimes \caH_{\Lambda\setminus \Lambda'})\right) = 1.
    \]
    The frustration-free property implies $\psi_\Lambda(R) \in \caG_{\Lambda'}\otimes \caH_{\Lambda\setminus \Lambda'}$, which establishes \eqref{eq:non_RMM_basis} by Lemma~\ref{lem:obc_gss}.

      For $R\in\caR_\Lambda^{MM}$, write $R= (\tilde{R}, M_n^{(i)})$ as in \eqref{eq:factored_bvmd} and let $R_D := (\tilde{R}, M_{n-2}^{(3)}, D_i)$. The two distinct roots that result from truncating any $T\leftrightarrow R$ to $\Lambda'$ are
    \begin{equation} \label{eq:restricted}
            R':=R\restriction_{\Lambda'} = (\tilde{R}, M_{n-1}^{(i)}), \quad  R_D':=R_D\restriction_{\Lambda '} = (\tilde{R}, M_{n-2}^{(3)}, R^{(i)})
    \end{equation}
    where $R^{(1)} = (V)$, $R^{(2)}= (V, M_1)$ and $R^{(3)} = (B_r)$. Thus, by Proposition~\ref{prop:counting_gs},
    \begin{equation}\label{eq:dim_RMM}
            \dim(\caC_\Lambda(R)\cap (\caG_{\Lambda'}\otimes \caH_{\Lambda\setminus \Lambda'})) = 2,
    \end{equation}
and it is easy to show given \eqref{eq:restricted} that
\begin{align*}
  \psi_{\tilde{\Lambda}_n^{(i)}}(\tilde{R})\otimes\varphi_{n-1}\otimes \ket{M_i} & = \psi_{\Lambda'}(R')\otimes\ket{R\restriction_{[L-2,L]}}  \\
   \psi_{\tilde{\Lambda}_n^{(i)}}(\tilde{R})\otimes\varphi_{n-2}\otimes \ket{D_i}  & = \psi_{\Lambda'}(R_D')\otimes\ket{R_D\restriction_{[L-2,L]}}.
\end{align*}
 These belong to $\caG_{\Lambda'}\otimes\caH_{\Lambda\setminus\Lambda'}$ by Lemma~\ref{lem:obc_gss}, and so \eqref{eq:RMM_basis} holds by \eqref{eq:eta} and \eqref{eq:dim_RMM}.
\end{proof}

So far, we have shown how to decompose each $\caC_\Lambda^\#(R)$, $\#\in\{{\rm per}, \, \cdot\, \}$, into invariant subspaces of $H_{\Lambda'}$ for any $\Lambda'\subseteq \Lambda$. One can take the direct sum of this over all $R\in \caR_{\Lambda}$ to produce a decomposition of $\caC_\Lambda^\#$. This can be simplified, though, by noticing that, for OBC, any root tiling $R'\in R_{\Lambda'}$ can be extended by zeros to a tiling $T\in\caT_\Lambda$. The same holds in the case of PBC as long as $|\Lambda|\geq |\Lambda'|+4$. Then, by similar arguments as in \eqref{eq:OBC_fragmentation}-\eqref{eq:PBC_fragmentation}
\begin{align}
    \caC_\Lambda & = \bigoplus_{R'\in\caR_{\Lambda'}}\bigoplus_{\substack{T\in\caT_\Lambda : \\ T\restriction_{\Lambda'}=R'}} \ket{T\restriction_{\Lambda_l'}} \otimes \caC_{\Lambda'}(R') \otimes \ket{T\restriction_{\Lambda_r'}}, \\
    \caC_\Lambda^{\rm per} & = \bigoplus_{R'\in\caR_{\Lambda'}} \bigoplus_{\substack{T\in\caT_\Lambda^{\rm per}: \\ T\restriction_{\Lambda'}=R'}} \caC_{\Lambda'}(R') \otimes \ket{T\restriction_{\Lambda\setminus \Lambda'}}, \quad\quad\text{if }\; |\Lambda|\geq |\Lambda'|+4.
\end{align}
As a consequence, if $\caC_{\Lambda'}(R')$ is invariant under $A\in \cB(\caH_{\Lambda'})$ for all $R\in\caR_{\Lambda'}$, the norm and spectrum of $A\restriction_{\caC_\Lambda^{\#}}$ can be calculated from that of $A\restriction_{\caC_{\Lambda'}}.$

\begin{lemma}\label{lem:isospectral}
Fix two intervals $\Lambda'\subseteq \Lambda$, and suppose that $A\in \cB(\caH_{\Lambda'})$ leaves $\caC_{\Lambda'}(R)$ invariant for all $R\in \caR_{\Lambda'}$. Then
\begin{enumerate}
    \item $\|A\otimes \idty_{\Lambda\setminus \Lambda'}\|_{\caC_\Lambda} = \|A\|_{\caC_{\Lambda'}}$ where the norm is as in \eqref{eq:restricted_norm}.
    \item $\spec(A\otimes\idty_{\Lambda\setminus \Lambda'}\restriction_{\caC_\Lambda}) = \spec(A\restriction_{\caC_{\Lambda'}}). $
\end{enumerate}
The same result holds after replacing $\caC_\Lambda$ with $\caC_\Lambda^{\rm per}$ as long as $|\Lambda|\geq |\Lambda'|+4.$
\end{lemma}
This result is an immediate consequence of the fact that the restriction of $A$ to either $\caC_\Lambda^\#$ or $\caC_{\Lambda'}$ is block diagonal with respect to the same set of invariant blocks. The only difference is that a block may appear multiple times in the decomposition of $\caC_\Lambda^\#.$ However, this does not change the norm or spectrum of the operator, only the multiplicity of the eigenvalues.

\section{Applying the bulk gap strategy to the Haldane pseudopotential}\label{sec:gap_proof}

\subsection{A lower bound on $E_1(\caV_\Lambda)$}\label{sec:E_1}

We now prove the lower bound on $E_1(\caV_\Lambda)$ for the truncated $1/3$-filled Haldane pseudopotential on the cylinder where $\caV_\Lambda= \caC_\Lambda^{\rm per}$. This follows from combining Theorem~\ref{thm:Knabe_FQHE} and Theorem~\ref{thm:MM_FQHE}. We start by applying the finite size criterion from Corollary~\ref{cor:fsc}. 

\begin{thm}\label{thm:Knabe_FQHE}
  Let $\Lambda = [1,L]$ be a ring of $L\geq 3n+9$ sites for some fixed $n\geq 2$. Then,
  \begin{equation}\label{eq:E_1_fsc}
      E_1(\caC_\Lambda^{\rm per}) \geq \frac{n}{2(1+2|\lambda|^2)(n-1)}\left(
      \min_{3\leq m \leq 5}\gap\left(H_{[1,3n+m]}\restriction_{\caC_{[1,3n+m]}}\right) - \frac{\kappa(1+2|\lambda|^2)}{n}
      \right).
  \end{equation}
\end{thm}

\begin{proof}
Take $r\in \{3,4,5\}$ and $N\in\bN$ so that $L=3N+r$, and define a sequence of intervals for $1\leq j \leq N+1$ by
\[
X_j =\begin{cases}
    [3j-2,3j+3], & 1\leq j \leq N \\
    [L-r+1,L+3] & j = N+1
\end{cases}
\]
where addition is understood modulo $L$. This collection of intervals satisfy conditions (ii)-(iii) of Corollary~\ref{cor:fsc} since (1) they all have at least 6 sites, (2) any two consecutive intervals intersect on three sites, and (3) any interaction term is supported on an interval of at most 4 sites. Moreover, condition (i) holds since $\caC_\Lambda^{\rm per}$ is invariant under every interaction term that comprises $H_\Lambda^{\rm per}.$ Hence, Corollary~\ref{cor:fsc} applies. Since $\caG_\Lambda^{\rm per}\subseteq \caC_\Lambda^{\rm per}$, \eqref{E_defs} and \eqref{eq:subspace_gap} imply that $\gap(H_\Lambda^{\rm per}\restriction_{\caC_\Lambda^{\rm per}})=E_1(\caC_\Lambda^{\rm per})$. Combining all of this with Lemma~\ref{lem:isospectral} results in the lower bound
\begin{equation}\label{eq:almost_result}
    E_1(\caV_\Lambda) \geq  \frac{\gamma n}{2C(n-1)}\left(
      \min_{3\leq m \leq 5}\gap\left(H_{[1,3n+m]}\restriction_{\caC_{[1,3n+m]}}\right) - \frac{C}{n}
      \right)
\end{equation}
where, since $6\leq r+3\leq 8$, the constants $\gamma$ and $C$ can be taken as
\begin{equation}\label{eq:fsc_app}
    \gamma := \min_{6\leq k\leq 8}\gap(H_{[1,k]}\restriction_{\caC_{[1,k]}})=\kappa, \quad
C := \max_{6\leq k \leq 8} \|H_{[1,k]}\|_{\caC_{[1,k]}}= \kappa(1+2|\lambda|^2).
\end{equation}

The explicit values of $\gamma$ and $C$ are a consequence of writing
\[
H_{[1,k]}\restriction_{\caC_{[1,k]}} = \bigoplus_{R\in\caR_{[1,k]}}H_{[1,k]}\restriction_{\caC_{[1,k]}(R)}
\]
and calculating $\spec(H_{[1,k]}\restriction_{\caC_{[1,K]}(R)})$ for all possible $R$. This is simplified from noting that, since $k\leq 8$, the spectrum is the same for any $R$ with the same number of consecutive monomers. If $R$ has no pairs of consecutive monomers, then $\caC_{[1,k]}(R)$ is the one-dimensional span of a ground state of $H_{[1,k]}$. If $R$ has two or three consecutive monomers, then $\caC_{[1,k]}(R)$ is a two-dimensional or three-dimensional invariant subspace of $H_{[1,k]}$, respectively. The spectrum can be easily calculated in both cases resulting in \eqref{eq:fsc_app}. Inserting \eqref{eq:fsc_app} into \eqref{eq:almost_result} produces \eqref{eq:E_1_fsc}.
\end{proof}

The lower bound on $E_1(\caC_\Lambda^{\rm per})$ has now been reduced showing that
\[
 \min_{3\leq m \leq 5}\gap\left(H_{[1,3n+m]}\restriction_{\caC_{[1,3n+m]}}\right) > \frac{\kappa(1+2|\lambda|^2)}{n} \qquad \text{for some}\;\; n \geq 2.
\]
This is accomplished in the next result where we use the martingale method to prove a lower bound on $\gap(H_{[1,L]})$ that is independent of $L.$

\begin{thm}\label{thm:MM_FQHE}
   Define $f(r) := \sup_{n\geq 4}f_n(r)$ for $r\geq 0$ where, recalling $\beta_n$ from \eqref{eq:beta},
    \begin{equation}
        f_n(r) := r\beta_{n}\beta_{n-2}\left(\frac{\left[1-\beta_{n-1}(1+r)\right]^2}{1+2r}+\beta_{n-3}\frac{r(1-\beta_{n-1})^2}{1+r}\right), \quad \forall \, r\geq 0.
    \end{equation}
   If $\lambda \neq 0$ and $L\geq 10$, the spectral gap of $H_{[1,L]}\restriction_{\caC_{[1,L]}}$ is bounded below by
    \begin{equation}
        \label{eq:E_1_MM}
        \gap(H_{[1,L]}\restriction_{\caC_{[1,L]}}) \geq \frac{\kappa}{3}\left(1-\sqrt{3f(|\lambda|^2)}\right)^2.
    \end{equation}
\end{thm}

The proof of this result follows closely the arguments given in \cite[Theorem 3.1]{Warzel:2022}-\cite[Lemma 3.2]{Warzel:2022}. However, we provide the details as this is a nontrivial step for proving Theorem~\ref{thm:main_result}.

\begin{proof}
We prove that Assumption~\ref{assump:MM} holds. To begin, let $N\geq 3$ and $r \in\{1,2,3\}$ denote the unique integers such that $L=3N+r$. Define the sequence of intervals
\begin{equation}
    X_n =
    \begin{cases}
        [1,6+r] & n=1 \\
        [3n+r-5, 3n+r+3] & 2\leq n\leq N-1
    \end{cases} \, .
\end{equation}
These satisfy $|X_n\cap X_{n+1}|=6$ for all $n$, $|X_1|=6+r$, and $|X_n|=9$ for $n\geq 2$. As such, Assumption~\ref{assump:MM}(iii) holds with $\ell=3$, and all interaction terms are supported on at least one and at most three of these intervals. Hence, Assumption~\ref{assump:MM}(ii) holds with $d=3$. Moreover, applying Lemma~\ref{lem:isospectral} and arguing similarly to \eqref{eq:fsc_app}, Assumption~\ref{assump:MM}(i) holds with
\[
\gamma := \inf_{m=7,8,9}\gap(H_{[1,m]}\restriction_{\caC_{[1,L]}}) = \inf_{m=7,8,9}\gap(H_{[1,m]}\restriction_{\caC_{[1,m]}}) = \kappa.
\]
If Assumption~\ref{assump:MM}(iv) holds with $\epsilon \leq f(|\lambda|^2)$, then \eq{eq:E_1_MM} will immediately follow from invoking Corollary~\ref{cor:MM}. The remained of this proof is focused on verifying this bound on $\epsilon$.

Recall that $\Lambda_n = \cup_{k=1}^n X_k$ and $E_n = G_{\Lambda_n}-G_{\Lambda_{n+1}}$. Since $\caC_{[1,L]}$ is invariant under $G_{\Lambda_n}$ for all $n$, \eqref{eq:invariant_relation} and Lemma~\ref{lem:isospectral} imply that
\[
\|G_{X_{n+1}}\|_{E_n\caC_{[1,L]}} = \|G_{X_{n+1}}(G_{\Lambda_n}-G_{\Lambda_{n+1}})\|_{\caC_{\Lambda_{n+1}}}, \quad \forall \, 0\leq n \leq N-2.
\]
This norm is trivially zero for $n=0$. Let $\Lambda =[1,M]$, $\Lambda' = [1, M-3]$ and $X = [M-8,M]$ for any $M\geq 10$. As the model is translation invariant, the desired bound on $\epsilon$ is immediate from showing 
\begin{equation}\label{eq:effective_epsilon}
    \|G_{X}(G_{\Lambda'}-G_{\Lambda})\|_{\caC_\Lambda}^2 = \sup_{\substack{0\neq \psi \in \caC_\Lambda: \\ \psi \in \caG_\Lambda^\perp\cap (\caG_{\Lambda'}\otimes \caH_{\Lambda\setminus \Lambda'})}}\frac{\|G_X\psi\|^2}{\|\psi\|^2} \leq f(|\lambda|^2).
\end{equation}
Note the equality above uses $G_{\Lambda'}-G_{\Lambda}= (\idty-G_{\Lambda})G_{\Lambda'}$ as $\caG_{\Lambda}\subseteq \caG_{\Lambda'}\otimes \caH_{\Lambda\setminus\Lambda'}$ by frustration-freeness.

By Lemma~\ref{lem:ob_tilings}, the subspace the supremum is take over in \eqref{eq:effective_epsilon} is
\[
 \{\psi\in\caC_\Lambda : \psi \in \caG_\Lambda^\perp\cap (\caG_{\Lambda'}\otimes \caH_{\Lambda\setminus\Lambda'})\}= \Span\{\eta_\Lambda(R) : R\in \caR_\Lambda^{MM}\} =:\cK_{\Lambda}.
\]
We first calculate $\|G_X\psi\|^2$ for any $\psi=\eta_\Lambda(R)$, and use this to bound the norm for an arbitrary $\psi\in\cK_\Lambda.$ Writing $\eta_\Lambda(R)$ as in \eqref{eq:factored_bvmd}, the first step is divided into two cases: $\Lambda_n^{(i)}\subseteq X$ and $X\subset \Lambda_n^{(i)},$ where we recall $\Lambda_n^{(i)}:=\Lambda\setminus\tilde{\Lambda}_n^{(i)}$ is the last $3(n-1)+i$ sites in $\Lambda.$ 

If $\Lambda_n^{(i)}\subseteq X$, the frustration-free property implies
\[
G_X\eta_\Lambda(R) = G_X\left(\psi_{\tilde{\Lambda}_n^{(i)}}(\tilde{R})\otimes G_{\Lambda_n^{(i)}}\eta_n^{(i)}\right) = 0,
\]
where the last equality holds since $\varphi_n^{(i)},\eta_n^{(i)}\in \caC_{\Lambda_n^{(i)}}(M_n^{(i)})$ are orthogonal by \eqref{eq:eta}, and $\eta_n^{(i)}$ is orthogonal to all other BVMD states by Lemma~\ref{lem:obc_gss}.

In the case $X\subset \Lambda_n^{(i)}$, the ground state projection $G_X$ can be expanded in terms of the orthogonal basis from Lemma~\ref{lem:obc_gss}, resulting in
\[
G_X\eta_\Lambda(R) = \psi_{\tilde{\Lambda}_n^{(i)}}(\tilde{R})  \otimes \left(\sum_{R\in\caR_X }\frac{\ketbra{\psi_X(R)}}{\|\psi_X(R)\|^2}\eta_n^{(i)}\right).
\]
Since $\eta_n^{(i)}\in\caC_{\Lambda_n^{(i)}}(M_n^{(i)})$, by \eqref{eq:OBC_fragmentation} the summation can further be reduced to roots $R\in \caR_{X}$ such that $R= T\restriction_X$, for some $T\leftrightarrow M_n^{(i)}.$ As $R' := M_n^{(i)}\restriction_X$ separates the particles of one pair of neighboring monomers, there are two such roots, $R'$ and $R_D' := D_n^{(i)}\restriction_X$ where $D_n^{(i)}=(M,\ldots, M, D, M, M_i).$ Using \eqref{eq:phi_1}-\eqref{eq:phi_3} to rewrite $\psi_X(R'), \psi_X(R_D')$, and $\eta_n^{(i)}$, one can easily calculate
\[
G_X\eta_n^{(i)} = \frac{\overline{\lambda}(1-\beta_{n-1}\|\varphi_2\|^2)}{\|\varphi_3\|^2}\varphi_{n-3}\otimes\varphi_3^{(i)}+\frac{|\lambda|^2(1-\beta_{n-1})}{\|\varphi_2\|^2}\varphi_{n-4}\otimes\ket{D}\otimes\varphi_2^{(i)}.
\]
Since $\|\eta_n^{(i)}\|^2=\|\varphi_{n-3}\|^2/(\beta_{n}\beta_{n-2})$ by \eqref{eq:beta}, applying $\|\varphi_k\|^2 = 1+(k-1)|\lambda|^2$, $k=2,3$, one obtains
\[
\|G_X\eta_\Lambda(R)\|^2 = f_n(|\lambda|^2)\|\psi_{\tilde{\Lambda}_n^{(i)}}(R)\|^2\|\eta_n^{(i)}\|^2 = f_n(|\lambda|^2)\|\eta_\Lambda(R)\|^2.
\]

As $\caC_\Lambda(R)$ is invariant under $G_X$ for all $R\in\caR_\Lambda$, the mutual orthogonality of the BVMD spaces from Lemma~\ref{lem:obc_gss} implies $\{G_X\eta_\Lambda(R) : R\in\caR_\Lambda^{MM}\}$ is an orthogonal set. Therefore, for an arbitrary $\psi=\sum_{R\in\caR_\Lambda^{MM}}c_R\eta_\Lambda(R)\in \cK_\Lambda$,
\[
\|G_X\psi\|^2 = \sum_{R\in\caR_\Lambda^{MM}}|c_R|^2\|G_X\eta_\Lambda(R)\|^2 \leq \sup_{n\geq 4}f_n(|\lambda|^2)\sum_{R\in\caR_\Lambda^{MM}}|c_R|^2\|\eta_\Lambda(R)\|^2 = f(|\lambda|^2)\|\psi\|^2,
\]
where we use that $X\subseteq \Lambda_n^{(i)}$ if and only if $n\geq 4.$ This establishes \eqref{eq:effective_epsilon} and completes the proof.
\end{proof}

\subsection{A lower bound on $E_0(\caV_\Lambda^\perp)$}\label{sec:E_0}

Consider a fixed finite volume $\Lambda =[1,L]$ in the ring geometry with $L\geq 10$. The goal of this section is prove the necessary lower bound on $E_0(\caV_\Lambda^{\perp})$ from \eqref{E_defs}. We summarize the ideas and calculations for the proof of Theorem~\ref{thm:GSE}, and point the reader to \cite[Theorem 3.4]{Warzel:2022} for the complete details.
\begin{thm}\label{thm:GSE}
    Let $\Lambda = [1,L]$ with $L\geq 11$ and set $\caV_\Lambda = \caC_\Lambda^{\rm per}$. Then
    \begin{equation}\label{eq:E_0_bound}
        E_0(\caV_\Lambda^\perp) \geq \frac{1}{3}\min\left\{1, \frac{\kappa}{\kappa+2}, \frac{\kappa}{2+2\kappa|\lambda|^2}\right\}.
    \end{equation}
\end{thm}

By the injectivity of $\sigma_\Lambda^{\rm per},$ the orthogonal complement of $\caV_\Lambda = \caC_\Lambda^{\rm per}$ can be written as the following span of configuration states:
\[
\caV_\Lambda^\perp = \left(\caC_\Lambda^{\rm per}\right)^\perp= \Span \{\ket{\mu} : \mu\in \cS_\Lambda\}, \qquad \caS_\Lambda = \{0,1\}^{|\Lambda|}\setminus \ran \sigma_\Lambda^{\rm per},
\]
The electrostatic energy $e_\Lambda(\mu) = \sum_{x=1}^L\mu_x\mu_{x+2}$ of any $\mu\in \caS_\Lambda$ is key for producing the lower bound in Theorem~\ref{thm:GSE}. By Proposition~\ref{prop:tiling_configs}, $\caS_\Lambda$ can be partitioned into the three subsets as:
\begin{align*}
   \caS_\Lambda^{(1)} &= \left\{\mu \in \{0,1\}^{|\Lambda|}: e_\Lambda(\mu) \geq 1 \right\} \\
  \caS_\Lambda^{(2)} &= \left\{\mu \in \{0,1\}^{|\Lambda|}: \mu_x=\mu_{x+1}=\mu_{x-4}=\mu_{x-5}=1\; \text{for some}\; x\in \Lambda\right\}\setminus \caS_\Lambda^{(1)} \\
  \caS_\Lambda^{(3)} &= \left\{\mu \in \{0,1\}^{|\Lambda|}: \mu_x=\mu_{x+1}=1, \; \mu_{x-3}+\mu_{x+4}\geq 1\; \text{for some}\; x\in \Lambda\right\}\setminus (\caS_\Lambda^{(1)}\cup \caS_\Lambda^{(2)})
\end{align*}

Suppose that $\psi=\sum_{\mu\in\caS_\Lambda}\psi(\mu)\ket{\mu}\in \left(\caC_\Lambda^{\rm per}\right)^\perp$ is arbitrary. The lower bound in \eqref{eq:E_0_bound} is a result of identifying constants $c_i\in\bN$ and $\gamma_i>0$ that are independent of $\Lambda$ such that
\begin{equation}\label{eq:gs_goal}
    c_i\braket{\psi}{H_\Lambda^{\rm per}\psi} = c_i\Big(\sum_{\mu \in \caS_\Lambda^{(1)}}e_\Lambda(\mu)|\psi(\mu)|^2 + \kappa \sum_{\nu\in \{0,1\}^{|\Lambda|}}\sum_{x\in\Lambda}|\braket{q_x^*\nu}{\psi}|^2 \Big)\geq \gamma_i\sum_{\mu\in \caS_\Lambda^{(i)}}|\psi(\mu)|^2.
\end{equation}
Since $\|\psi\|^2 = \sum_{\mu\in\caS_\Lambda}|\psi(\mu)|^2$, dividing both sides of \eqref{eq:gs_goal} by $c_i$, and summing over $i=1,2,3$ yields
\begin{equation}\label{eq:gs_bound_idea}
  E_0((\caC_\Lambda^{\rm per})^\perp) \geq \frac{1}{3}\min\left\{\frac{\gamma_i}{c_i}: i=1,2,3\right\}.  
\end{equation}

Clearly, \eqref{eq:gs_goal} holds for $i=1$ with $\gamma_1=c_1=1$. For $i=2,3$, any $\mu\in\cS_\Lambda^{(i)}$ can be connected to some $\eta\in\caS_\Lambda^{(1)}$ using at most two hopping terms $q_x^*q_x.$ The argument that establishes \eqref{eq:gs_goal} for $i=2,3$ combines this observation with the following Cauchy-Schwarz bound:
\begin{equation}\label{eq:CS}
|a+b|^2 \geq (1-\delta)|a|^2 -\frac{1-\delta}{\delta}|b|^2,\qquad \forall a,b\in \bbC,\; \delta\in(0,1)  \,.  
\end{equation}

For example, when $\mu\in\caS_\Lambda^{(3)}$, take $\eta\in\caS_\Lambda^{(1)}$ and $\caD_\mu=\{(\nu,x)\}\subset\{0,1\}^{|\Lambda|}\times \Lambda$ as follows:
\[x\equiv x(\mu):= \max\{y\in[1,L] : \mu_{y}=\mu_{y+1}=1\, \wedge \,\mu_{y-3}+\mu_{y+4}\geq 1\},\]
$\nu=\nu(\mu)$ is given by removing the particles at $x$ and $x+1$ from $\mu$, and $\eta=\eta(\mu)$ is given by hopping the particles at $x$ and $x+1$ in $\mu$ to $x-1$ and $x+2$. Since $\mu\notin\caS_\Lambda^{(1)}$, $\eta$ is well-defined, and $e_\Lambda(\eta)=\eta_{x-3}\eta_{x-1}+\eta_{x+2}\eta_{x+4}\geq 1$. Applying \eqref{eq:CS} with $\delta = \tfrac{\kappa|\lambda|^2}{1+\kappa|\lambda|^2}$ produces
   \begin{equation}\label{eq:one_config}
        e_\Lambda(\eta) + \kappa\sum_{(\nu,x)\in D_\mu}|\braket{q_x\nu}{\psi}|^2 = e_\Lambda(\eta)+\kappa|\psi(\mu)-\overline{\lambda}\psi(\eta)|^2 \geq  \frac{\kappa}{1+\kappa|\lambda|^2}|\psi(\mu)|^2:=\gamma_3|\psi(\mu)|^2.
   \end{equation}
   
   The definition of $x(\mu)$ guarantees $D_\mu\cap D_{\mu'}=\emptyset$ if $\mu\neq \mu'$. Hence, summing \eqref{eq:one_config} over all $\mu\in\caS_\Lambda^{(3)}$ shows that \eqref{eq:gs_goal} holds with $\gamma_3$ and
    \[c_3=\max_{\eta\in\caS_\Lambda^{(1)}}|\{\mu\in \caS_\Lambda^{(3)}:\eta(\mu)=\eta\}|.\] The assumption $|\Lambda|\geq 11$ guarantees $\mu\mapsto \eta(\mu)$ is one-to-one when $e_\Lambda(\eta(\mu))=2$. There are at most two $\mu\neq \mu'\in \caS_\Lambda^{(3)}$ such that $\eta(\mu)=\eta(\mu')$ when $e_\Lambda(\eta(\mu))=1$.  Thus, $c_3=2$.

 A similar calculation holds for $\mu\in \caS_\Lambda^{(2)}$ after appropriately choosing $\eta\in\caS_\Lambda^{(1)}$ and $\caD_\mu =\{(\nu_i,x_i):i=1,2\}$. In this case, applying \eqref{eq:CS} twice with well-chosen $\delta_1,\delta_2\in(0,1)$ produces
 \[
 e_\Lambda(\eta)+\kappa \sum_{(\nu,x)\in\caD_\mu}|\braket{q_{x}^*\nu}{\psi}|^2 \geq \frac{\kappa}{\kappa+2}|\psi(\mu)|^2 :=\gamma_2|\psi(\mu)|^2.
 \]
 The configuration $\eta$ and set $D_\mu$ can be taken so that $\caS_\Lambda^{(2)}\ni\mu\mapsto\eta(\mu)$ is one-to-one, and $D_\mu\cap D_{\mu'}=\emptyset$ when $\mu\neq \mu'$. This yields $c_2=1.$ See \cite[Theorem 3.4]{Warzel:2022} for more details.

\subsection{Proof of Theorem~\ref{thm:main_result}} Let $\Lambda=[1,L]$ with $L\geq 11.$ Combining Theorems~\ref{thm:Knabe_FQHE}-~\ref{thm:MM_FQHE} produces the following lower bound on the spectral gap of $H_\Lambda^{\rm per}\restriction_{\caC_\Lambda^{\rm per}}$:
\[
E_1(\caC_\Lambda^{\rm per}) \geq \frac{\kappa n}{2(1+2|\lambda|^2)(n-1)}\left(\frac{1}{3}\left(1-\sqrt{3f(|\lambda|^2)}\right)^2-\frac{1+2|\lambda|^2}{n}\right) \quad \forall\; n \leq L/3-3.
\]
This is positive for $n \geq (3+6|\lambda|^2)/(1-\sqrt{3f(|\lambda|^2)})^2$ and independent of the system size for all $L\geq 3n+9.$ Moreover, by Theorem~\ref{thm:GSE},
\[
E_0((\caC_\Lambda^{\rm per})^\perp) \geq \gamma^{\rm per}:=\frac{1}{3}\min\left\{1, \frac{\kappa}{2+2\kappa|\lambda|^2}, \frac{\kappa}{2+\kappa}\right\}
\]
As $H_\Lambda^{\rm per}\caC_\Lambda^{\rm per}\subseteq \caC_\Lambda^{\rm per}$ and $\caG_\Lambda^{\rm per}\subseteq\caC_\Lambda^{\rm per}$, by \eqref{E_defs} these two bounds yields the desired result:
\[
\liminf_{L\to\infty}\gap(H_\Lambda^{\rm per}) \geq \min\left\{\frac{\kappa}{6(1+2|\lambda|^2)}\left(1-\sqrt{3f(|\lambda|^2)}\right)^2, \, \gamma^{\rm per}\right\}.
\]

\subsection{The edge state energy} The lower bounds on $E_1(\caV_\Lambda)$ and $E_0(\caV_\Lambda^\perp)$ have a positive limit as $|\lambda|\to 0$. This implicitly implies that the edge states of the model with OBC belong to $\caV_\Lambda^\perp.$ To show this explicitly, we analyze the ground state energy of $H_\Lambda\restriction_{\caC_\Lambda^\perp}.$ This analysis runs identically to that of Theorem~\ref{thm:GSE} with one main change. Namely, the interaction terms $q_x^*q_x$, $x=1,L-1,L$, cannot be used in the analysis as they are not interaction terms of $H_\Lambda$, see \eqref{eq:OBC_Ham}. 

The sets $\caS_\Lambda^{(i)}$ can be analogously defined in the OBC case using the convention that any site outside $\Lambda$ is vacant. However, we further partition $\caS_\Lambda^{(3)}$ by whether $D_\mu=\{(\nu,x)\}$ in \eqref{eq:one_config} could have been chosen with $x\notin\{1,L-1,L\}$. Hence, define
\[
\caS_{\Lambda^o}^{(3)} = \left\{\mu \in \caS_{\Lambda}^{(3)} :\exists x\in [2,L-2] \;\text{s.t.}\; \mu_x=\mu_{x+1}=1, \, \mu_{x-3}+\mu_{x+4}\geq 1 \right\}, \quad \caS_{\partial\Lambda}^{(3)} = \caS_\Lambda^{(3)}\setminus \caS_{\Lambda^o}^{(3)}.
\]
We point out that $\caS_{\partial\Lambda}^{(3)}= \{\sigma_\Lambda(R):R\in\caR_\Lambda^{e}\}$ where $\caR_\Lambda^{e}$ is as in Section~\ref{sec:edge_states}.

The analysis for $\mu \in \caS_{\Lambda^o}^{(3)}$ follows that of $\mu \in \caS_\Lambda^{(3)}$ from \eqref{eq:one_config} after restricting $x(\mu)\in[2,L-2]$. A similar adjustment ensures that no $q_x^*q_x$, $x=1,L-1,L$, is used to analyze $\mu\in\caS_\Lambda^{(2)}$. It is only $\mu\in\caS_{\partial\Lambda}^{{(3)}}$ where further adjustments are needed. We give the argument for when $\mu_1=\mu_2=\mu_5=1$. The case that $\mu_{L-4}=\mu_{L-1}=\mu_L=1$ runs analogously. 

Let $\caD_\mu=\{(\nu,x)\}$ where $x=3$ and $\nu$ the configuration that removes the particles from sites $2$ and $5$ in $\mu$, and let $\eta$ be the configuration that hops the particles at sites $2$ and $5$ in $\mu$ to sites $3$ and $4$. Then applying \eqref{eq:CS} with $\delta = \frac{\kappa}{\kappa+1}$ produces
\[
e_\Lambda(\eta) + \kappa\sum_{(\nu,x)\in\caD_\mu}|\braket{q_x^*\nu}{\psi}|^2\geq \frac{\kappa|\lambda|^2}{\kappa+1}|\psi(\mu)|^2.
\]
Then, arguing analogously as in \eqref{eq:gs_goal}-\eqref{eq:gs_bound_idea} produces
\[
\min_{0\neq \psi\in \caC_\Lambda^\perp }\frac{\braket{\psi}{H_\Lambda\psi}}{\|\psi\|^2} \geq \frac{1}{4}\min\left\{1,\,  \frac{\kappa}{\kappa+2}, \, \frac{\kappa}{2+2\kappa|\lambda|^2}, \, \frac{\kappa|\lambda|^2}{\kappa+1} \right\}
\]
which is clearly $\mathcal{O}(|\lambda|^2)$ in the limit $|\lambda|\to 0.$


 \section*{Acknowledgements} This project was funded by the DFG under TRR 352 (Grant: Number: 470903074):  Mathematik der Vielteilchen-Quantensysteme und ihrer kollektiven Phänomene. The author would like to thank Bruno Nachtergaele and Simone Warzel for helpful comments on a draft of this work.

\bibliography{Qmath_22}

\bibliographystyle{plain}

\end{document}